\documentclass[10pt, numbers]{sigplanconf}
\usepackage[utf8]{inputenc} 
\usepackage[T1]{fontenc}
\usepackage{microtype}
\usepackage{amssymb,amsmath}
\usepackage{enumerate,paralist,enumitem}
\usepackage[ruled,vlined]{algorithm2e}
\usepackage{lipsum}
\usepackage{soul}
\usepackage{color}
\usepackage[lofdepth,lotdepth,caption=false]{subfig}
\usepackage{multicol}
\usepackage{tabularx}
\usepackage{url}
\usepackage{xargs}
\usepackage{adjustbox}
\usepackage{multirow}
\usepackage[table,xcdraw]{xcolor}
\usepackage{etoolbox}
\usepackage{amsthm}
\usepackage{tikz}
\usetikzlibrary{tikzmark}

\newcommand*{\mytt}{\fontfamily{lmtt}\selectfont}
\newcommand*{\code}[1]{{\mytt #1}}

\newtheorem{defn}{Definition}
\newtheorem{thm}{Theorem}

\usepackage[colorinlistoftodos,prependcaption,textsize=tiny]{todonotes}
\newcommandx{\unsure}[2][1=]{\todo[linecolor=red,backgroundcolor=red!25,bordercolor=red,#1]{#2}}
\newcommandx{\change}[2][1=]{\todo[linecolor=blue,backgroundcolor=blue!25,bordercolor=blue,#1]{#2}}
\newcommandx{\info}[2][1=]{\todo[linecolor=OliveGreen,backgroundcolor=OliveGreen!25,bordercolor=OliveGreen,#1]{#2}}
\newcommandx{\improvement}[2][1=]{\todo[linecolor=Plum,backgroundcolor=Plum!25,bordercolor=Plum,#1]{#2}}
\newcommandx{\thiswillnotshow}[2][1=]{\todo[disable,#1]{#2}}

\newcommand\mc{\multicolumn}

\newcommand\set[1]{\{#1\}}

\newcommand\mx{\sqcup}
\newcommand\cle{\sqsubseteq}

\newcommand\cnflct{\asymp}

\newcommand{\ids}{=^\sig}

\newcommand{\lcp}{\prec_{\textit{CP}}}
\newcommand{\lcps}{\prec_{\textit{CP}}^\sig}
\newcommand{\cp}{\le_{\textit{CP}}}
\newcommand{\cps}{\le_{\textit{CP}}^\sig}

\newcommand{\nlw}{\nprec_{\textit{WCP}}}
\newcommand{\lwcp}{\prec_{\textit{WCP}}}
\newcommand{\wcp}{\le_{\textit{WCP}}}
\newcommand{\wcps}{\le_{\textit{WCP}}^\sig}
\newcommand{\lws}{\prec_{\textit{WCP}}^\sig}

\newcommand{\hb}{\le_{\textit{HB}}}
\newcommand{\rel}[1]{{\mathtt{rel(#1)}}}
\newcommand{\acq}[1]{{\mathtt{acq(#1)}}}
\newcommand{\sync}[1]{{\mathtt{sync(#1)}}}
\newcommand{\acrl}[1]{{\mathtt{acrl(#1)}}}
\newcommand{\cs}{\textup{CS}}
\renewcommand\l{\mathtt{l}}

\newcommand{\hbs}{\le_{\textit{HB}}^\sig}

\newcommand{\trs}{<_{\textit{tr}}^\sig}
\newcommand{\tr}{<_{\textit{tr}}}
\newcommand{\tres}{\le_{\textit{tr}}^\sig}
\newcommand{\tre}{\le_{\textit{tr}}}
\newcommand{\tos}{<_{\textit{TO}}^\sig}
\newcommand{\tho}{<_{\textit{TO}}}
\newcommand{\toes}{\le_{\textit{TO}}^\sig}
\newcommand{\toe}{\le_{\textit{TO}}}

\newcommand{\mtc}{\textit{match}}

\newcommand{\W}[1]{{\mytt w(#1)}}
\newcommand{\R}[1]{{\mytt r(#1)}}

\newcommand\proj{\upharpoonright}

\newcommand\sig{\sigma}

\renewcommand\H{\mathbb{H}}
\renewcommand\P{\mathbb{P}}
\renewcommand\L{\mathbb{L}}
\newcommand{\cN}{\mathbb{N}}
\newcommand\C{\mathbb{C}}
\newcommand{\Rc}{\mathbb{R}}
\newcommand{\Wc}{\mathbb{W}}

\newcommand{\tool}{{\textsc{Rapid}}}
\newcommand{\rvpredict}{$\mathsf{RVPredict}$}

\newtoggle{techreport}
\toggletrue{techreport}

\usepackage{booktabs}

\title{Dynamic Race Prediction in Linear Time}
\authorinfo{Dileep Kini\and Umang Mathur\and Mahesh Viswanathan}
 	{University of Illinois at Urbana-Champaign, USA}
 	{dileeprkini@gmail.com, umathur3@illinois.edu, vmahesh@illinois.edu}
\usepackage{flushend}

\begin{document}

\toappear{}

% \newdimen\origiwspc%
%   \newdimen\origiwstr%
%   \origiwspc=\fontdimen2\font% original inter word space
%   \origiwstr=\fontdimen3\font% original inter word stretch
%   \fontdimen3\font=0.1em% inter word stretch

\maketitle

\begin{abstract} Writing reliable concurrent software remains a huge challenge
for today's programmers. Programmers rarely reason about their code by
explicitly considering different possible inter-leavings of its execution.  We
consider the problem of detecting data races from individual executions in a
sound manner. The classical approach to solving this problem has been to use
Lamport's happens-before (HB) relation. Until now HB remains the only approach
that runs in linear time. Previous efforts in improving over HB such as
causally-precedes (CP) and maximal causal models fall short due to the fact that
they are not implementable efficiently and hence have to compromise on their
race detecting ability by limiting their techniques to bounded sized fragments
of the execution. We present a new relation weak-causally-precedes (WCP) that
is provably better than CP in terms of being able to detect more races, while
still remaining sound. Moreover it admits a linear time algorithm which works on the
entire execution without having to fragment it. \end{abstract}

% 2012 ACM Computing Classification System (CSS) concepts
% Generate at 'http://dl.acm.org/ccs/ccs.cfm'.



%!TEX root = main.tex

\section{Introduction}\label{sec:intro}

Writing reliable concurrent program remains a huge challenge;
depending on the order in which threads are scheduled, there are a
large number of possible executions. Many of these executions remain
unexplored despite extensive testing. The most common symptom of a
programming error in multi-threaded programs is a data race. A data
race is a pair of \emph{conflicting} memory accesses such that in some
execution of the program, these memory accesses are performed
consecutively; here, by conflicting memory accesses, we mean, pair of
read/write events to the same memory location performed by different
threads, such that at least one of them is a write. The goal of
dynamic race detectors is to discover the presence of a data race in a
program by examining a \emph{single} execution. Given its singular
role in debugging multi-threaded programs, dynamic race detection has
received robust attention from the research community since the
seminal papers~\cite{lamport1978time,savage1997eraser} more than two decades
ago.

All dynamic race detection algorithms can be broadly classified into
three categories. First are the \emph{lock-set} based
approaches~\cite{savage1997eraser,elmas2007goldilocks} that detect \emph{potential} data
races by tracking the set of locks held during each data access. These
methods are fast and have low overhead, but are \emph{unsound} in that
many potential races reported are spurious. The second class of
techniques falls in the category of \emph{predictive runtime analysis}
techniques~\cite{rv2014,Said2011}. Here the race detector explores all possible
reorderings of the given trace, searching for a possible witness that
demonstrates a data race. These techniques are precise --- races
detected are indeed data races, and they are likely to find all such
races. The downside is that they are expensive. A single trace has
potentially exponentially many reorderings. Therefore, these
techniques are applied by slicing the trace in small-sized fragments,
and searching for a race in these short fragments. The last class of
techniques are what we call \emph{partial order} based techniques. In
these, one identifies a partial order $P$ on the events in the trace
such that events unordered by $P$ correspond to ``concurrent
events''. These algorithms are sound (presence of unordered
conflicting events indicates a data race) and have low overhead
(typically polynomial in the size of the trace). However, they are
conservative and may miss anomalies detected by the predictive runtime
analysis techniques. The approach presented in this paper falls in
this last category.

Happens-before (HB)~\cite{lamport1978time} is the simplest, and most commonly used
partial order to detect races. It orders the events in a trace as
follows: (i) two events performed by the same thread are ordered the
way they appear in the trace, and (ii) synchronization events across
threads are also HB-ordered in order of appearance in the trace if
those events access the same synchronization objects. 
Rule (i) says that we cannot reorder events within a thread because we have no
information about the underlying program which allows us to infer an
alternate execution of the thread. Rule (ii) says that we cannot
reorder synchronization events on the same objects (in our case locks)
as they would lead to violation of mutual exclusion (critical sections
on same lock should not overlap---also referred to as \emph{lock
semantics}). Consider for example the trace shown in
Figure~\ref{fig:no-swap}. (In all the example figures we follow the
convention of representing events of the trace top-to-bottom, where
temporally earlier events appear above the later ones. We also use the
syntax of $\acq{l}/\rel{l}$ for acquire/release events of lock
$\tt{l}$ and \R{x}/\W{x} for read/write events on variable $\tt{x}$.)
The events on lines 4 and 5 cannot be interchanged temporally as
mutual exclusion will be violated.  We could consider circumventing
this lock semantics violation by reordering the entire critical
sections (which would also change relative positions of events on
lines 4, 5), but we cannot infer such a move because the \R{x} event
of thread $t_2$ could see a different value which could cause
alternative executions in the underlying program and hence the events following
it might be different.  So in this case the
HB reasoning, of avoiding lock semantics violation correctly, though
unwittingly, prevented the swapping of critical sections (the swapping
does not violate lock semantics). Had the two \W{x} events been absent,
we could have actually swapped the two critical sections temporally to
get a feasible alternate execution. For example, the trace in
Figure~\ref{fig:swap} can be reordered to expose a race on the access
of $\tt{y}$ by performing the critical section of $t_2$ before the
other. Such a race is called \emph{predictable} as the trace that
exposes it can be obtained from the given trace by rearranging the
temporal order of events across threads. 
That is, it can
be predicted from the trace without having to look at the underlying
program.  For Figure~\ref{fig:swap}, HB will still not
declare a race since events on lines 4 and 5 would be ordered.

\begin{figure}[t]
\centering
\subfloat[cannot swap critical sections]{\label{fig:no-swap}
	\begin{tabular}{c}
		{\mytt
		\begin{tabularx}{.2\textwidth}{r|XX|}
		\cline{2-3}
		&\mc{1}{c}{$t_1$} & \mc{1}{c|}{$t_2$}\\
		\cline{2-3}
		1 & acq(l) &\\
		2 & r(x) &\\
		3 & w(x) &\\
		4 & rel(l) &\\
		
		5 && acq(l)\\
		6 && r(x)\\
		7 && w(x)\\
		8 && rel(l)\\
		\cline{2-3}
		\end{tabularx}}\\
		% \\
		% \begin{tabular}{rr|}
		% Ordering & Result\\
		% \hline
		% {\it CP}  & ``no race''\\
		% {\it WCP} & ``race''\\
		% \end{tabular}
	\end{tabular}
}
% \rulesep
\subfloat[][can swap critical sections]{\label{fig:swap}
	\begin{tabular}{c}
		{\mytt
		\begin{tabularx}{.2\textwidth}{r|XX|}
		\cline{2-3}
		&\mc{1}{c}{$t_1$} & \mc{1}{c|}{$t_2$}\\
		\cline{2-3}
		1 & w(y) &\\
		2 & acq(l) &\\
		3 & r(x) &\\
		4 & rel(l) &\\
		
		5 && acq(l)\\
		6 && r(x)\\
		7 && rel(l)\\
		8 && r(y)\\
		\cline{2-3}
		\end{tabularx}}\\
		% \\
		% \begin{tabular}{rr|}
		% \quad Ordering & \quad Result\\
		% \hline
		% {\it CP}  & ``race''\\
		% {\it WCP} & ``race''
		% \end{tabular}
	\end{tabular}
}
\caption{Example traces showing when critical sections can/cannot be swapped. }
\label{fig:hb}
\end{figure}
% Lockset techniques \cite{} were then introduced to detect races such
% as those presented in Figure~\ref{fig:swap}, by checking that traces
% follows \emph{locking discipline}: where access to shared variables
% are required to be within the scope of some common lock otherwise a
% race is declared. This enables it to detect a race condition on the
% events on $\tt{y}$ in Figure~\ref{fig:swap} as they are not protected
% by any synchronization events.  But this reasoning is unsound as it
% would lead to declaring races when there are none. Consider the trace
% in Figure~\ref{fig:unpredict} (borrowed from \cite{cp2012}) where
% access to variable $\tt{y}$ does not follow locking discipline and
% hence lockset declares a race. But there is no reordering of the
% events that reveals any race because any feasible reordering would
% require performing the \R{x} after the \W{x} to make sure that when
% $\tt{x}$ is read it sees the same value and hence the same execution
% can be followed as in the given trace, because there is no guarantee
% that the \R{y} will be performed if the \R{x} saw a different
% value. This makes lockset based methods unsound which means they can
% declare spurious races which are a hindrance for the programmer.
The partial order Causally-Precedes (CP)~\cite{cp2012} was introduced
to detect races missed by HB such as those in Figure~\ref{fig:swap}
while being sound. The CP relation is a subset of HB, which implies
that CP can detect races above and beyond those detected by HB. 
Soundness of CP guarantees that a CP race is either an actual race or a deadlock. 
% The salient feature of CP is that it is sound, i.e., 
% if CP detector issues a warning then there is either a race or a deadlock. 
But there are two
main drawbacks of CP. Firstly CP misses races that are
predictable. Consider the traces shown in Figure~\ref{fig:cp}; the
only difference between the two traces is that lines 6 and 7 have been
swapped. There is no predictable race in Figure~\ref{fig:unpredict}
because \R{x} in $t_2$ (line 6) must be performed after \W{x} in $t_1$
(line 3), which prevents the critical sections from being
reordered. On the other hand, Figure~\ref{fig:predict} has a
predictable race on $\tt{y}$ --- the sequence $e_5,e_6,e_1$ reveals
the race ($e_i$ refers to the event at line $i$). CP, however, does
not detect a race in either trace, because it is agnostic to the
ordering of events within a critical section. The
second drawback of CP is that, while it can be detected in polynomial
time~\cite{cp2012}, there is no known linear time algorithm~\footnote{We
  believe that there is quadratic time lower bound on \emph{any} CP
  algorithm.}. This severely hampers its use
on real-world examples with traces several gigabytes large. 
So any implementation of CP must resort to \emph{windowing} 
where the trace is partitioned into small fragments.
% like the predictive runtime verification techniques. 
This means that it can only find races within
bounded fragments of the trace, and so detects fewer races than what CP
promises. Our experiments (Section~\ref{sec:experiments}) reveal
windowing to be a serious impediment to race detection in large
examples.
A recent implementation of CP~\cite{raptor2016} performs an online analysis
 while avoiding windowing,
even though, theoretically, its running time is not linear. 
Currently, it seems slower than our implementation of WCP, 
as it processes roughly a few million events in a few hours.

%\todo{where is correct reordering?}CR\\

We address the two drawbacks of CP in one shot. We
propose a partial order Weak Causally-Precedes (WCP) which, as the
name suggests, is a weakening of the CP partial order. Thus WCP
detects all races that CP does and even more (like the race in
Figure~\ref{fig:predict} as explained in
Section~\ref{sec:illustrations}). We prove that WCP enjoys the same
soundness guarantees as CP. Additionally, like HB, WCP admits a linear
time Vector-Clock algorithm for race detection, thus solving the main
open problem proposed in~\cite{cp2012}. This is surprising because
when HB was weakened to obtain CP it resulted in detecting more races
but came at the cost of an expensive algorithm. But weakening CP to
WCP not only allows for detecting more races but enables an efficient
algorithm.

\begin{figure}[t]
\centering
\subfloat[no predictable race]{\label{fig:unpredict}
	\begin{tabular}{c}
		{\mytt
		\begin{tabularx}{.2\textwidth}{r|XX|}
		\cline{2-3}
		& \mc{1}{c}{$t_1$} & \mc{1}{c|}{$t_2$}\\
		\cline{2-3}
		1 & w(y) &\\
		2 & acq(l) &\\
		3 & w(x) &\\
		4 & rel(l) &\\
		
		5 && acq(l)\\
		6 && r(x)\\
		7 && r(y)\\
		8 && rel(l)\\
		\cline{2-3}
		\end{tabularx}}\\
		% \\
		% \begin{tabular}{rr|}
		% Ordering & Result\\
		% \hline
		% {\it CP}  & ``no race''\\
		% {\it WCP} & ``race''\\
		% \end{tabular}
	\end{tabular}
}
%\rulesep
\subfloat[][predictable race]{\label{fig:predict}
	\begin{tabular}{c}
		{\mytt
		\begin{tabularx}{.2\textwidth}{r|XX|}
		\cline{2-3}
		& \mc{1}{c}{$t_1$} & \mc{1}{c|}{$t_2$}\\
		\cline{2-3}
		1 & w(y) &\\
		2 & acq(l) &\\
		3 & w(x) &\\
		4 & rel(l) &\\
		
		5 && acq(l)\\
		6 && r(y)\\
		7 && r(x)\\
		8 && rel(l)\\
		\cline{2-3}
		\end{tabularx}}\\
		% \\
		% \begin{tabular}{rr|}
		% \quad Ordering & \quad Result\\
		% \hline
		% {\it CP}  & ``race''\\
		% {\it WCP} & ``race''
		% \end{tabular}
	\end{tabular}
}
\caption{Example traces showing how CP misses race due to small change.}
\label{fig:cp}
\end{figure}

Our key contributions are the following:
\vspace{-0.01in}
\begin{enumerate}
\item We define a new \emph{sound} relation weak-causally-precedes (WCP),
  which is weaker than causally-precedes. The definition and its
  soundness proof, as claimed in~\cite{cp2012}, is challenging.
\begin{quote}
``It is worth emphasizing that multiple researchers have fruitlessly
  pursued such a weakening of HB in the past. \ldots (Both the
  definition of CP and our proof of soundness are results of
  multi-year collaborative work, with several intermediate failed
  attempts.)''
\end{quote}
Even though we had the benefit of following the work on CP, our
experience concurs with the above observation. The subtlety in these
ideas is highlighted by the fact that the soundness proof for CP,
presented in~\cite{cp2012}, is \emph{incorrect}; informed readers can
find an explanation of the errors in the CP soundness proof in
\iftoggle{techreport}{Appendix~\ref{sec:cperrors}}{our companion technical report~\cite{techreport}}. 
Our attempts at fixing the CP proof led us to
multiple years of fruitless labor until finally the results presented
here. Our soundness proof for WCP (which also, by definition, applies
to CP) requires significant extensions to the proof ideas outlined
in~\cite{cp2012} (see \iftoggle{techreport}{Appendix~\ref{sec:wcpsound}}{\cite{techreport}} for the full proof).

\item We achieve the holy grail for dynamic race detection algorithms
  --- a \emph{linear} running time. It is based on searching for
  conflicting events that are unordered by WCP. We prove that our
  algorithm is \emph{correct}. 
  Our algorithm uses linear space in the worst case, 
  as opposed to the logarithmic space requirement of happens-before vector clock algorithm.
  However, in our experiments, we did not encounter these worst case bounds.
  Our algorithm scales to traces with hundreds of millions of events
  and the memory usage stays below $3\%$ for most benchmarks (see Table~\ref{tab:exp}).
  %   However, the worst case memory requirement does not bottleneck the performance
  % of the algorithm --- our algorithm 
  % scales to traces having hundreds of millions of events, with low memory overhead
  % (see Table~\ref{tab:exp}; the memory usage is below $3\%$ for most benchmarks).
  We further show that our algorithm is
  \emph{optimal} in terms of its asymptotic running time and memory
  usage.
  
{}
\item Our experiments show the benefits of our algorithm, in terms of
  the number of races detected and its efficiency when compared to
  state of the art tools such as \rvpredict{}. They reveal the
  tremendous power of being able to analyze the entire trace as opposed to
  trace fragments that other sound race detection algorithms (other than
  those based on HB) are forced to be restricted to.

\end{enumerate}

%The rest of the paper is organized as follows. 
In Section~\ref{sec:relation}
we describe the partial order WCP, and how it is a weakening of the CP relation. 
Section~\ref{sec:algo} describes the Vector-Clock algorithm that implements
WCP faithfully and runs in linear time. 
In Section \ref{sec:experiments} we describe the implementation and the experimental
results. We provide related work in
Section~\ref{sec:related} and give concluding remarks in Section~\ref{sec:conclusion}.

%!TEX root = main.tex

\section{Weak Causal Precedence}
\label{sec:relation}

\subsection{Preliminaries}

In this paper we consider the sequential consistency model for assigning
semantics to concurrent programs wherein the execution is viewed as an
interleaving of operations performed by individual threads. The
possible operations include lock acquire and release ($\acq{l}$, $\rel{l}$) and
variable access which include read and write to variables (\R{x}, \W{x}).

\paragraph{Orderings:} Let $\sig$ be a sequence of events. We say $e_1$ is
\emph{earlier than} $e_2$ according to $\sig$, denoted by $e_1 \trs e_2$, when
$e_1$ is performed before $e_2$ in the trace $\sig$. We shall use $t(e)$ to
denote the thread that performs $e$. We say $e_1$ is \emph{thread ordered
before} $e_2$, denoted by  $e_1 \tos e_2$ to mean that $e_1 \trs e_2$ and
$t(e_1) = t(e_2)$. We use $\ids$ to denote the identity relation on events of
$\sig$ and $\tres,\toes$ to denote the relations $(\trs\cup\ids)$ and
$(\tos\cup\ids)$.  We shall drop $\sig$ from the superscript of the relations
when it is clear from context. We use $\sig\!\proj_t$ to denote the projection
of $\sig$ onto the events by thread $t$.

% For an event $e$ we use $\sig(e)$ to denote the position (an integer) of the
% event $e$ within the trace $\sig$.

\paragraph{Lock events:} For a lock acquire/release event $e$, we use $l(e)$ to
denote the lock on which it is operating. For an acquire event $a$ we use
$\mtc(a)$ to denote the earliest release event $r$ such that $l(a) = l(r)$ and
$a \tho r$. We similarly define the match of a release event $r$ to be the
latest acquire $a$ such that $l(a) = l(r)$ and $a \tho r$.  A \emph{critical
section} is the set of events of a thread that are between (and including) an
acquire and its matching release, or if the matching release is absent then
all events that are thread order after an acquire. We use $\cs(e)$ to denote a
critical section starting/ending at an acquire/release event $e$.  For event
$e$ and lock $\ell$ we use $e {\in} \ell$ to denote that $e$ is contained in a
critical section over $\ell$.

% We also use a matching acquire/release pair $(a,r)$ to refer to the
% corresponding critical section.

% We say an acquire/release pair $(a,r)$ is a \emph{critical section} if $r =
% \mtc(a)$. For an acquire/release event $e$ we use $\cs(e)$ to denote the
% critical section starting/ending at $e$.

% We represent each event is of the form $\ev{t}{a}$ where $a$
% denotes the operation being performed and $t$ is the thread performing the
% operation.  We will often use it with a subscript $\ev{t}{a}_i$ to denote its
% position in the given sequence. 

\paragraph{Trace:} For a sequence of events $\sig$ to be called a \emph{trace} it needs to satisfy two
properties:
\begin{enumerate} \item \emph{lock semantics:} for any two acquisition events
$a_1$ and $a_2$ if $l(a_1) = l(a_2)$ and $a_1 \tr a_2$ then $r_1 = \mtc(a_1)$
exists in $\sig$ and $r_1 \tr a_2$
% $\ev{t_1}{acq(l)}_{i_1}$, $\ev{t_2}{\acq{l}}_{i_2}$ where $i_1 < i_2$
% there exists a release of that lock $\ev{t_1}{\rel{l}}_{j}$ where $i_1 < j < i_2$
\item \emph{well nestedness:} for any critical section $C$,
if there exists an acquire event $a$ such that $a \in C$,
then $r = \mtc(a)$ exists and $r \in C$.
% $\ev{t}{\acq{l'}}_{j_1}$ such that $i_1 < j_1 < i_2$ then there
% is a release $\ev{t}{\rel{l'}}_{j_2}$ such that $j_1 < j_2 < i_2$.
\end{enumerate}
\paragraph{Race:} Two events are said to be \emph{conflicting} if they access
the same variable and at least one of them is a write and the events  are
performed by different threads. We use $e_1 \asymp e_2$ to denote that $e_1$
and $e_2$ are conflicting.  When one can execute a concurrent program such
that conflicting events can be performed next to each other in the trace, we
say that the trace has revealed a \emph{race} in the program.

\paragraph{Deadlock:} When a program is being executed to obtain a trace, the
scheduler picks a thread and performs the ``next event'' in that thread. Note
that when this event is performed the next event of the other threads is not
going to be affected. This concept of next event is needed for understanding a
deadlock. A trace is said to reveal a \emph{deadlock} when a set of threads
$D$ cannot proceed because each of them is trying to acquire a lock that is
held by another thread in that set. In other words, the next event in each
thread in $D$ is an $\acq{l}$ such that the lock $\tt{l}$ is acquired by another
thread in $D$ without having released it. In such a situation none of the
threads in $D$ can proceed because lock semantics will be violated, and no
matter how the rest of the threads proceed the next event of these threads
will not change.

\paragraph{Predictability and Correct Reordering:} In order to formalize the concept of predictable
race/deadlock we need the notion of correct  reordering.  
% A trace $\sig'$ is said to be a \emph{correct reordering} of another 
% trace $\sig$ if for every thread $t$, $\sig'\!\proj_t$ is a prefix of $\sig\!\proj_t$, 
% and every read event in $\sig'$ sees/returns the same value as it did in $\sig$,
% i.e., the last \W{x} event before the \R{x} in the two sequences are the same.
A trace $\sig'$ is said to be a \emph{correct reordering} of another 
trace $\sig$ if for every thread $t$, $\sig'\!\proj_t$ is a prefix of $\sig\!\proj_t$, 
and the last \W{x} event before any \R{x} event is the same in both $\sig$ and $\sig'$.
This ensures that every read event in $\sig'$ sees/returns the same value as it did in $\sig$.
We say a trace $\sig$ has a \emph{predictable race} (\emph{predictable deadlock}) if there is a correct
reordering of it which exhibits a race (deadlock). 

% \hl{next event well defined?} \todo{examples of correct/incorrect reorderings}

\subsection{Partial Orders}

\newcommand{\E}{\mathcal{E}}
\newcommand{\po}{\le_P^{\sig}}
\newcommand{\pr}{\parallel}

Given a trace $\sig$, we consider various partial orders on its events.
Formally, let $\E$ be the set of events in $\sig$.
A partial order $P$ is a
binary relation $\po$ on $\E$ (i.e. $\po \subseteq \E{\times}\E$) 
which is reflexive, antisymmetric and transitive. {The relations $\tres,\toes$
defined earlier are examples of partial orders}. 
We say two events $e_1,e_2$
are unordered by a partial order $P$, denoted by $e_1 \pr_P e_2$, when neither
$e_1 \po e_2$, nor $e_2 \po e_1$.

% We will only be interested in partial orders $\po$
% that satisfy thread order: $\tos \subseteq \po$ for any $\sig$.

We say a trace $\sig$ exhibits a $P$-race between events $e_1,e_2$ when they
are conflicting ($e_1 \cnflct e_2$) and are unordered by $P$ ($e_1 \pr_P
e_2$). The aim, in this race-detection paradigm of using partial orders, is to
design partial orders $P$ which can guarantee that for any $\sig$, the presence
of $P$-race in $\sig$ implies the presence of a predictable-race in $\sig$.
Such partial orders are said to be $\emph{strongly sound}$. 
%There is also another version of soundness that has been considered. 
A partial order is said to be \emph{weakly sound} if
for any trace containing unordered conflicting events there is a correct
reordering of the trace which reveals either a race or a deadlock. From a
programmer's perspective deadlocks are as undesirable as races. 
% Therefore, aiming for weakly sound partial orders can be useful not only for detecting
% deadlocks but the weaker notion might allow for detection of additional races
% that otherwise might be harder to detect.  
Additionally, a weaker notion might allow for detection of additional races
that are otherwise harder to detect. 
For programs which are
guaranteed to be deadlock-free, the two notions coincide.

% A partial order is said to be \emph{strongly sound} if for any
% trace containing unordered (by the partial order) conflicting events there is
% a correct reordering of the given trace which reveals a race.  

% A sequence of events $\rho$ over $\E$ (a permutation of events in $\E$) is
% said to be a linearization of $\po$ if for any two events $e_1,e_2$ if $e_1
% \po e_2$  then $e_1 \tr^{\rho} e_2$. Each partial order for a given trace
% defines a set of associated linearizations of which the given trace is one.
% In this sense we think of a partial order as a generalization of the given
% trace. For instance the relation $\trs$ defined earlier is a partial order
% with exactly one linearization which is $\sig$.  This implies that
% linerizations preserve well nestedness.

Before we delve into specific partial orders, we present some intuition
regarding partial orders and alternative executions. Linearizations of a
partial order $\po$ are possible executions of the trace with different
interleavings of the threads. 
If two events are unordered by some partial order, then we can
obtain a linearization in which the two events are placed next to one another
(performed simultaneously). 
But this by itself does not imply soundness. 
This is because a linearization $\rho$ of $\po$ 
(i) might not be a valid trace as it might violate lock semantics, or 
(ii) might not be a correct reordering of $\sig$ as there might be a read event 
\R{x} whose corresponding last \W{x} event in $\rho$ does
not match with that in $\sig$. 
% This is because a linearization of $\po$ might not be a trace or might not be a
% correct reordering of the given trace $\sig$. Observe that a linearization
% $\rho$ of $\po$ may not be a correctly reordered trace of $\sig$ on two
% accounts: (i) $\rho$ is not a trace due to lock semantics being violated (ii)
% there is read event \R{x} whose corresponding last \W{x} event in $\rho$ does
% not match with that in $\sig$. 
Therefore, two conflicting events unordered by a
partial order only indicates the possibility of, and does not necessarily
guarantee the existence of a predictable race.
The cleverness lies in designing partial orders
for which this possibility is indeed a guarantee.

%  Hence the ability to compute such partial orders (ability to compare
% two events according the partial order) gives rise to race detection
% techniques, since a race is revealed by a trace in which conflicting events
% are performed simultaneously.

We will now describe partial orders HB, CP (from literature) and WCP
(our contribution).
 %   We define each of them using \emph{rules}
 % to derive {orderings} between events of the trace.

% \begin{defn}[Happens-Before] Given a trace $\sig$, $\hbs$ is the smallest partial order 
% on the events in $\sig$ such that $\tos \subseteq \hbs$ and for every release event $r$ 
% and acquire event $a$, if  $l(a) = l(r)$ and $r \trs a$, then $r \hbs a$.
\begin{defn}[Happens-Before] Given a trace $\sig$, $\hbs$ is the smallest partial order 
on the events in $\sig$ with $\tos \subseteq \hbs$ that satisfies the following rule:
for a $\rel{l}$ event $r$ and an $\acq{l}$ event $a$, if $r \trs a$, then $r \hbs a$.
% if $r$ is a $\rel{l}$ event and $a$ is an $\acq{l}$ event such that $r \trs a$, then $r \hbs a$.
% \begin{enumerate}[label=\textup{{\alph*)}},itemsep=3pt]
% \item thread ordered events are also HB ordered: if $e_1 \tos e_2$  then $e_1 \hbs e_2$
% \item 
% \item $\hbs$ is closed under composition with itself i.e. $\hbs = (\hbs\circ\hbs)$. In other words 
% if $e_1 \hbs e_2$ and $e_2 \hbs e_3$ then $e_1 \hbs e_3$.
% \end{enumerate}
\end{defn}

Two salient features of HB are: (i) it is strongly sound, and
(ii) it can be computed in linear time.
Several race detection tools \cite{Pozniansky:2003:EOD:966049.781529,fasttrack} have been
developed using this technique. 
However, as discussed in Section~\ref{sec:intro},
HB can miss many races.
The Causally-Precedes (CP)
partial order was then introduced \cite{cp2012} as an improvement over HB. CP
is a subset of HB i.e it has fewer orderings and  hence more possible
interleavings.
It is thus able to detect more races than HB. CP is
proved to be weakly sound. 
However, it is not known if CP can be
computed in linear time. 
The race detection algorithm for CP is
{polynomial} time but not linear. 
% The implementation of CP uses a windowing
% strategy to chunk the trace into bounded fragments on which the race detection
% algorithm is applied.

%We take a look at the definition of CP. 
% A point made in \cite{cp2012} is that
% CP is not a partial order because it is not reflexive. The original definition
% does not satisfy thread order either. We present an extended version
% to address these two issues but still exclude all the non-trivial orderings
% that the original defintion omits from HB.
\begin{defn}[Causally-Precedes]
Given a trace $\sig$, $\lcps$ is the smallest relation
satisfying the following rules:

\begin{enumerate}[label=\textup{{(\alph*)}},itemsep=4pt]

% \item if $r$ is a $\rel{l}$ event and $a$ is an $\acq{l}$ event such that $r \trs
% a$ and their critical sections contain conflicting events $(e_1 {\in} CS(r)$, $e_2
% {\in} CS(a)$, $e_1 \cnflct e_2)$ then $r \lcps a$.
\item 
for a $\rel{l}$ event $r$ and an $\acq{l}$ event $a$ with $r \trs a$,
if the critical sections of $r$ and $a$ contain conflicting 
events $(e_1 \in \cs(r)$, $e_2 \in \cs(a)$, $e_1 \cnflct e_2)$, 
then $r \lcps a$.

% \item if $r$ is a $\rel{l}$ event and $a$ is an $\acq{l}$ event such that $r \trs
% a$ and their critical sections contain CP-ordered events $(e_1 {\in} CS(r), e_2
% {\in} CS(a), e_1 \lcps e_2)$ then $r \lcps a$.
\item for a $\rel{l}$ event $r$ and an $\acq{l}$ event $a$ with $r \trs a$,
if the critical sections of $r$ and $a$ contain CP-ordered 
events $(e_1 \in \cs(r), e_2 \in \cs(a), e_1 \lcps e_2)$,
 then $r \lcps a$.

\item $\lcps = {(\lcps \circ \hbs)} = {(\hbs \circ \lcps)}$,
i.e., $\lcps$ is closed under composition with $\hb$.

%  If $e_1 \lcps e_2$, $e_2 \hb e_3$ and {$t(e_1) \neq t(e_2)$} then $e_1 \lcps e_3$
% and similarly if $e_1 \hbs e_2$, $e_2 \lcps e_3$ and $t(e_2) \neq t(e_3)$ then $e_1 \lcps e_3$
%\item thread ordered events are also CP ordered: if $e_1 \tos e_2$  then $e_1 \cps e_2$
% CP is closed under left and right composition with HB\\
% $\wcp \text{=} (\wcp\circ\hb) \text{=} (\hb\circ\wcp)$\\
\end{enumerate}
\end{defn}

% The CP defined above is not exactly the one define in \cite{cp2012}. The one
% we define is reflexive which makes it a partial order unlike the original. The
% one we define also satisfies thread order unlike the original defintion, \hl{this
% makes stating the correctness of the vector clock algorithm easier}

As seen in the examples in Figure~\ref{fig:cp},
CP is agnostic to the relative order of the events inside 
the same critical sections.
% One of the drawbacks of the CP relation is that it is agnostic to the
% relative order of the events inside the same critical sections as seen in the
% examples in Figure~\ref{fig:cp}. 
That is, if we were to consider two
read/write events that are enclosed within the same set of critical sections
in the trace, and we interchanged their positions, then the resulting trace
would have exactly the same CP orderings across threads. This is because Rule
(a) of CP, that depends upon the position of read/write events, is only
concerned with whether or not it occurs inside some critical section and not
how they are relatively order within the critical section. 
This constraint prevents CP it from detecting races, as we saw in
Figure~\ref{fig:predict}.

Next we look at WCP, the partial order we introduce in this paper. It is
obtained by weakening rules (a) and (b) of CP as follows:

\begin{defn}[Weak-Causally-Precedes] 

Given a trace $\sig$, $\lws$ is the smallest relation that
satisfies the following rules:

\begin{enumerate}[label=\textup{{(\alph*)}},itemsep=4pt]

% \item if $r$ is a $\rel{l}$ event and $e \in \l$ is a \R{x}/\W{x} event
%  such that $r \trs e$ and $\cs(r)$ contains an event $e'$ conflicting 
%  with $e$ $(e' \in \cs(r)$, $e \cnflct e')$ then $r \lws e$.
\item 
for a $\rel{l}$ event $r$ and a read/write event $e \in \l$ with $r \trs e$,
if $\cs(r)$ contains an event conflicting 
with $e$ $(e' \in \cs(r)$, $e \cnflct e')$,
then $r \lws e$.

% \item If $r_1$ and $r_2$ are two $\rel{l}$ events such that $r_1 \trs r_2$,
% and their critical sections contain WCP-ordered events, i.e.,
% ${(e_1 \in \cs(r_1)}$, $ {e_2 \in \cs(r_2)}$, $e_1 \lws e_2)$, then ${r_1 \lws r_2}$.
\item 
for $\rel{l}$ events $r_1, r_2$ with $r_1 \trs r_2$,
if the critical sections of $r_1, r_2$ contain WCP-ordered events
${(e_1 \in \cs(r_1)}$, $ {e_2 \in \cs(r_2)}$, $e_1 \lws e_2)$, 
%$({e_1 \in \cs(r_1)}, {e_2 \in \cs(r_2)}, {e_1 \lws e_2})$, 
then ${r_1 \lws r_2}$.

\item $\lws = {(\lws \circ \hbs)} = {(\hbs \circ \lws)}$, i.e.,
 $\lws$ is closed under composition with $\hbs$,

\end{enumerate}
\end{defn}

Note that, Rule (a) of WCP orders the release event before the
read/write event (and not the acquire as in CP). 
Thus, WCP makes a distinction between events of the same thread 
based on the relative order inside a critical section.

Note that $\lcps$ and $\lws$ are not partial orders as they are not reflexive.
And unlike $\hbs$ they do not contain thread order. 
But if we consider the relations $\cps = (\lcps \cup \toes)$ and $\wcps = (\lws \cup \toes)$, then both
$\cps$ and $\wcps$ are partial orders. When defining races we use these partial
orders. 
% We drop $\sig$ from superscript of the relations when it is not
% important or is clear from the context. 
When clear from the context, we drop $\sig$ from superscript of the relations.
Rules (a) and (b) in WCP are weaker
versions of rules (a) and (b) in CP respectively. One can prove inductively
that $\wcp \subseteq \cp$ hence any CP-race is also a WCP-race. Next we state the soundness theorem
for WCP whose proof is provided in \iftoggle{techreport}{the Appendix}{\cite{techreport}}.

\begin{thm}[Soundness of WCP]\label{thm:sound} 
WCP is weakly sound, i.e.,
given any trace $\sig$, if $\sig$ exhibits a WCP-race then $\sig$ exhibits a
predictable race or a predictable deadlock.  \end{thm}

\subsection{Illustrations}\label{sec:illustrations} Going back to the example
in Figure~\ref{fig:predict} let us see how WCP is able to detect a race that
CP cannot. Note that the $\rel{l}$ in $t_1$ is CP-ordered before the $\acq{l}$ in
$t_2$ by rule (a) of CP. Further using rule (c) of CP we obtain that the
operations on variable \code{y} are CP-ordered and therefore CP does not
detect the predictable race as uncovered by the trace $e_5,e_1,e_6$. WCP on
the other hand does not order the $\rel{l}$ and $\acq{l}$. Rule (a) of WCP only
orders the \R{x} in $t_2$ after the $\rel{l}$ of $t_1$. If we look at the
example to its left in Figure~\ref{fig:unpredict} the same reasoning can be
used to obtain that CP does not detect any race, and indeed there is no
predictable race. In the case of WCP since the \R{x} in $t_2$ appears above
the \R{y}, it ends up ordering \W{y} and \R{y} thus not declaring any race
either.

The intuitive reason for formulating rule (a) in CP is that 
% \emph{``two conflicting events have to occur in the same order in 
% every correctly reordered execution if it is the case that they 
% do not constitute a race.
% Consequently the critical sections containing them need 
% to be ordered in their entirety''}~\cite{cp2012}. 
``\emph{two conflicting events have to occur in the same order in 
every correctly reordered execution if it is the case that they 
do not constitute a race.
Consequently the critical sections containing them need 
to be ordered in their entirety}''~\cite{cp2012}. 
This intuition is correct when we know for sure
that the conflicting events \emph{do occur} in the correctly reordered
execution (a correctly reordered execution need not include all the events in
the given execution, but only prefixes for each thread). In the example in
Figure~\ref{fig:predict} the two operations on \code{x} do not appear in the
set of events that need to be scheduled in order to reveal a race, and hence
any orderings derived from them need not be considered. But then if the \R{x}
is indeed included in some (other) reordering then the $\rel{l}$ to $\acq{l}$
ordering should be respected. Rule (a) of WCP does not enforce this ordering
completely, it only makes sure that the later of the conflicting events is
ordered after the earlier release. It would seem that this ordering is not
sufficient to enforce lock semantics of the traces we would be interested in,
but it turns out that this weaker version is sufficient to guarantee
soundness.

We also weaken rule (b) in the same spirit as above. When considering critical
sections over the same lock containing ordered events, instead of ordering the
critical sections entirely (as CP does by ordering $\rel{l}$ of the earlier CS
and the $\acq{l}$ of the latter) WCP only orders  the earlier $\rel{l}$ before
the latter $\rel{l}$. 

\begin{figure}[h]
\centering
\begin{tabularx}{0.45\textwidth}{r|XXX|}
\cline{2-4}
 & \quad $t_1$ & \quad $t_2$ & \quad $t_3$\\
\cline{2-4}
1 & $\acq{l}$ & & \\
2 & $\sync{x}$ & &\\
3 & \R{z} & & \\
4 & $\rel{l}$ & & \\
5 & & $\sync{x}$ & \\
6 & & $\acq{l}$ & \\
7 & & $\acq{n}$ & \\
8 & & $\rel{n}$ & \\
9 & & $\rel{l}$ &\\
10 & & & $\acq{n}$\\
11 & & & $\rel{n}$\\
12 & & & \W{z}\\
\cline{2-4}
\end{tabularx}
\caption{Example to demonstrate how weakening rule (b) is useful. CP:``No race''. WCP:``Race''}\label{fig:ruleb}
\end{figure}

In Figure~\ref{fig:ruleb} we show a trace in which there
is a predictable race which is not detected by CP but is detected by WCP owing
to weakening of rule (b). We use $e_i$ to refer to event at line number $i$.
Borrowing notation from \cite{cp2012}, the event $\sync{x}$ is a shorthand for
$\acq{x}$ \R{xVar}\W{xVar} $\rel{x}$, where $\tt{xVar}$ is the unique variable
associated with lock $\tt{x}$.  Any two $\texttt{sync}$ events over the same
lock are ordered by CP/WCP using Rule (a) (for WCP the ordering is from the
$\rel{x}$ to the latter \R{xVar} but for our example this is not important). Now
we describe the example in Figure~\ref{fig:ruleb} in detail. First note that
there is a predictable race between $e_3$ and $e_{12}$ as revealed by the
correctly reordered trace $e_1,e_2,e_{10},e_{11},e_3,e_{12}$. This race is not
detected by CP as the conflicting events $e_3$ and $e_{12}$ are CP related as
follows: Firstly $e_2 \lcp e_5$. We have $e_1 \hb e_2$ and $e_5 \hb e_6$
(thread ordering). Applying Rule (c) of HB/CP composition on $e_1 \hb e_2 \lcp
e_5 \hb e_6$ we get $e_1 \lcp e_6$. Next we apply Rule (b) of CP on the
critical sections over lock $\tt{l}$ to get $e_4 \lcp e_6$. We also have $e_8
\hb e_{10}$ (critical sections over same lock $\tt{n}$). Applying Rule (c) of
CP ($e_3 \tho e_4 \lcp e_6 \tho e_8 \hb e_{10} \tho e_{12}$) we obtain $e_3 \lcp
e_{12}$.  We see that this line of reasoning fails if you try to prove $e_3
\lwcp e_{12}$ because $e_4$ and $e_6$ are not $\lwcp$ related. Rule (b) of WCP
would give us that the releases are ordered i.e., $e_4 \lwcp e_9$ but this
cannot be composed with HB edges to get a path to $e_{12}$. It turns out $e_3$
and $e_{12}$ are not ordered by $\lwcp$, and therefore the race that exists
between them is correctly detected by WCP.

\begin{figure}[ht]
\begin{tabularx}{0.45\textwidth}{r|XXX|}
\cline{2-4}
& \quad $t_1$ & \quad $t_2$ & \quad $t_3$\\
\cline{2-4}
1 & $\acq{l}$ & &\\
2 & $\acq{m}$ & &\\
3 & $\rel{m}$ & &\\
4 & \R{z} & &\\
5 & $\rel{l}$ & &\\
6 & & $\acq{m}$ &\\
7 & & $\acq{n}$ &\\
8 & & $\sync{x}$ &\\
9 & & $\rel{n}$ &\\
10 & & $\rel{m}$ &\\
11 & & & $\acq{n}$\\
12 & & & $\acq{l}$\\
13 & & & $\rel{l}$\\
14 & & & $\sync{x}$\\
15 & & & \W{z}\\
16 & & & $\rel{n}$\\
\cline{2-4}
\end{tabularx}
\caption{Example trace exhibiting a predictable race that is detected by WCP but not CP}\label{fig:ex3}
\end{figure}

We present a more involved example in Figure~\ref{fig:ex3} to once again see
how a predictable race is revealed by WCP but not CP. Starting with the two
$\sync{x}$ events that are CP related we get enclosing critical section on lock
$\tt{n}$ to be CP related by Rule (b) i.e., $e_9 \lcp e_{11}$. We then use $e_3
\hb e_6$ along with thread orderings and Rule (c) to get $e_1 \lcp e_{12}$.
Applying Rule (b) we obtain $e_5 \lcp e_{12}$ which then gives $e_{4} \lcp
e_{15}$ from Rule (c), hence CP declares that the two conflicting events
\R{z}/\W{z} are not in race. But they are indeed in race as revealed by the
correctly reordered trace
$e_6,e_7,e_8,e_9,e_{10},e_{11},e_{12},e_{13},e_{14},e_1,e_2,e_3,e_4,e_{15}$. 
Unlike CP, WCP does not
order $e_4$ and $e_{15}$. This is because WCP does not order $e_9$ and
$e_{11}$ since the weaker Rule (b) only gives $e_9 \lwcp e_{14}$ which does not
compose with the HB edge on the critical sections over lock $\tt{l}$.

% Now let us see why this examples gives indications that a single pass
% algorithm for CP might be hard. Any such algorithm assigns
% timestamps to the different events as it processes the 

\begin{figure}[h]
\begin{tabularx}{0.45\textwidth}{r|XXX|}
\cline{2-4}
&\quad $t_1$ &\quad $t_2$ &\quad $t_3$\\
\cline{2-4}
1 & $\acq{l}$ & &\\
2 & $\acq{m}$ & &\\
3 & $\rel{m}$ & &\\
4 & \R{z} & &\\
5 & $\rel{l}$ & &\\
6 & & $\acq{m}$ &\\
7 & & $\acq{n}$ &\\
8 & & $\sync{x}$ &\\
9 & & $\rel{n}$ &\\
10 & & & $\acq{n}$\\
11 & & & $\acq{l}$\\
12 & & & $\rel{l}$\\
13 & & & $\sync{x}$\\
14 & & & \W{z}\\
15 & & & $\rel{n}$\\
16 & & & $\sync{y}$\\
17 & & $\sync{y}$ &\\
18 & & $\rel{m}$ &\\
\cline{2-4}
\end{tabularx}
\caption{Example trace exhibiting a predictable deadlock but no predictable race}\label{fig:ex4}
\end{figure}

In Figure~\ref{fig:ex4} we show a trace that is only slightly different from
the previous example but the subtlety involved results in the absence of a
predictable race but presence of a predictable deadlock. The reordered trace
$e_1,e_6,e_{10}$ exhibits the deadlock.  Using identical reasoning as in the
previous example we can derive that the two conflicting events \R{z}/\W{z} are
CP ordered and WCP unordered. The proof of correctness of CP shows that it
cannot detect deadlocks involving more than 2 threads as in this example. This
shows that WCP can detect deadlocks that CP cannot.

%\todo{deadlocks detected by WCP and not CP (2 threads)}

%!TEX root = main.tex

\section{Vector-Clock Algorithm for WCP}~\label{sec:algo}
\newcommand\Nt{{\mathcal{T}}}
\newcommand\Nl{\mathcal{L}}
\newcommand\Nv{\mathcal{V}}
\newcommand\Nr{\mathcal{R}}
\newcommand\Nc{\mathcal{C}}
\newcommand\Ne{\mathcal{N}}

In this section we describe a vector clock algorithm that \emph{implements}
WCP. The algorithm assigns a ``timestamp'' $C_e$ to each event $e$.  
These timestamps are vector times which can be compared to each other. The key
property about the timestamps is that it preserves the WCP relation ($\wcp$).
This means in order to find out if two events are WCP ordered we can simply
compare their timestamps. Therefore we also refer to these timestamps as WCP
time or simply time.

\subsection{Vector Clocks and Times}

Let us first recall some basic notions pertaining to vector times. 
A vector time \textit{VT : Tid $\to$ Nat}, is a function that maps each
thread in a trace to a non-negative integer. 
It can also be viewed as $\Nt$-tuple, where $\Nt$ 
is the number of threads in the given trace. 
Vector times support comparison operation $\cle$ 
for point-wise comparison, 
join operation ($\mx$) for point-wise maximum,
and component assignment of the form $V[t := n]$
which assigns the time $n {\in} $\textit{Nat} to component $t {\in}$\textit{Tid}
 of vector time $V$. 
Vector time $\bot$ maps all threads to 0.
\\

\begin{tabular}{rcl}
$V_1 \cle V_2$ & iff & $\forall t: V_1(t) \le V_2(t)$\\
$V_1 \mx V_2$ & = & $\lambda t: \mathit{max}(V_1(t), V_2(t))$\\
$V[u := n]$ & = & $\lambda t: \mathtt{if}\; (t = u)\; \mathtt{then}\; n\; \mathtt{else}\; V(t)$\\
$\bot$ & = & $\lambda t: 0$\\
\end{tabular}
\\
% \begin{align*}
% V_1 \cle V_2 & \text{ iff } & \forall t: V_1(t) \le V_2(t)\\
% V_1 \mx V_2 & = & \lambda t: \mathit{max}(V_1(t), V_2(t))\\
% V[u := n] & = & \lambda t: \mathtt{if}\; (t = u)\; \mathtt{then}\; n\; \mathtt{else}\; V(t)\\
% \bot & = & \lambda t: 0\\
% \end{align*}

%  by their timestamps then they are ordered by WCP
% partial order (Theorem~\ref{thm:algorithm}). In other words the partial
% order on the events induced by the timestamps is isomorphic to the WCP partial
% order modulo . The salient feature of the algorithm is that it runs in
% $O(n\log{n})$ time and space where $n$ is the length of the trace.

% Let us recall some notions pertaining to vector clocks.    $\bot$ corresponds to
% the minimal element with all threads mapped to $0$. $V[t:=c]$ represents the
% vector clock with thread $t$ being mapped to $c \in \nat$ and the rest of them
% are the same as $V$ does.

Before we describe the algorithm we would like to point out the distinction
between \emph{clocks} and \emph{times}. Clocks are to be thought of as
variables, they are place holders for times which are the values taken up by
the clock. The time of a clock will change as the trace is processed. 
Events will be assigned different kinds of times based on the value of different
clocks right after the event is processed.
We use double struck font for denoting clocks (e.g., $\C, \P, \H, \cN$)
and normal font for denoting vector times (e.g., $C, P, H, N$).

\subsection{Algorithm}

Our algorithm works in a streaming fashion and processes the
trace by looking at its events one-by-one from beginning to end. As it handles
each event it updates its \emph{state}. The state captures all the information
required to assign a timestamp (a vector time) to the last event of the
trace. At each step, an event is processed and the state is updated, which
allows us to compute a timestamp of that event from the updated state.  For
each thread $t$ the state is going to consist of a vector clock $\C_t$ (among
other things) that reflects the time of the last event in the thread $t$ so
far. The timestamp/time of an event $e$ denoted by $C_e$ is simply the value of
$\C_{t(e)}$ just after having processed $e$. To ensure that the assigned times
preserve WCP ordering, the algorithm needs to ensure that the time of an event
$a$ is \emph{``communicated''} to $b$ if there is a WCP edge from $a$ to $b$.
Thus, $b$ can be assigned a time $C_b$ such that $C_a \cle C_b$. 
% This communication of time requires us to record 
% the time of certain past events
% (events in the prefix of the trace that has been processed)
% in the state using vector clocks.
In order to achieve this communication, we refer to
auxiliary times $P_e$ (\emph{WCP-predecessor time}) 
and $H_e$ (\emph{HB time}) associated with every event $e$.
Formally $P_e = \bigsqcup \{C_{e'} \;|\; e'\lwcp e\}$ and
$H_e =  \bigsqcup \{C_{e'} \;|\; e'\hb e\}$.
We use vector clocks $\P_t$ and $\H_t$ in the state to record times
$P_e$ and $H_e$ of the last events $e$ of thread $t$.
% \\
% \begin{center}
% \begin{tabular}{ccccccc}
% $P_e $& $=$ & $\bigsqcup\limits_{e'\lwcp e} C_{e'}$ & \quad&
% $H_e $& $=$ & $\bigsqcup\limits_{e'\hb e} C_{e'}$\\
% \end{tabular}
% \end{center}
% \begin{alignat*}{6}
% P_e & = & \bigsqcup\limits_{e'\lwcp e} C_{e'} & \quad
% H_e & = & \bigsqcup\limits_{e'\hb e} C_{e'}
% \end{alignat*}

In the following paragraphs we motivate different components of the state,
and explain how these components are updated as part of
the algorithm. The state update procedure is described in
Algorithm~\ref{algo:update}. 
It consists of procedures that are prompted depending on the type of
event being processed. Each procedure has an argument $t$ denoting the thread that is
performing the event. The argument $\ell$ for acquire/release events
denotes the lock being operated. The argument $x$ for read/write denotes the
variable being accessed in the event. The parameter $L$ for read/write events denotes
the set of locks corresponding to the enclosing critical sections of the
event. The parameters $R/W$ for the release event correspond to the set of variables
that have been read/written inside the critical section corresponding to the
release in question. The purpose of these parameters will become clear as we
move forward. 
Next we describe different components of the state along with
how these procedures manipulate them.%\todo{remove this for space}

%\subsection{Components of the State}\label{sec:state}

\paragraph{Local Clock $\cN_t$ :}With each thread $t$ we associate an integer
counter $\cN_t$ in the state. It represents the \emph{local clock} of thread
$t$. The \emph{local time} of an event $e$, denoted by $N_e$, is the value of $\cN_t$
after $e$ is processed. For any $t$ the value $\C_t(t)$ will always refer to
$\cN_t$. %The local clock $\cN_t$ is initialized to 1.

\paragraph{Local Clock Increment :} The local clock $\cN_t$ is incremented just
before an event of $t$ is processed iff the previous event in $t$ was a
release. Since this increment is common to all events, we omit it from the
pseudocode.

% The key property we will maintain is that for
% any two events $a,b$ where $t(a) \neq t(b)$ if $N_a \le C_b(t(a))$ then $a
% \lwcp b$. In simpler words, if $b$ has the knowledge of $a$'s local clock then
% $a$ should precede $b$.

\paragraph{WCP clocks $\P_t,\P_\ell$ :} With each thread we associate a vector
clock $\P_t$. The \emph{WCP-predecessor time} $P_b$ 
of an event $b$ is the value of the clock $\P_{t(b)}$ after $b$ has been
processed.  
As defined earlier, $P_b$ represents the ``knowledge'' that $b$ has about other events
with respect to the $\lwcp$ relation. 
Note that, for any event $a$, if $N_a \le P_b(t(a))$ 
(i.e., $P_b$ ``knows'' $a$), then it is the case that $a \lwcp b$ 
(\iftoggle{techreport}{Lemma~\ref{lem:clockimplieswcp} in Appendix}{Lemma C.8 in~\cite{techreport}}). 
This invariant is maintained by making sure that whenever $a \lwcp b$ then the
WCP-time of $a$, $C_a$ (which also has local time of $a$), is made known to
$b$ so that it can update $\P_{t(b)}$ appropriately. 
Sending the entire vector time $C_a$
(rather than just $N_a$) to $b$ ensures that 
$P_b$ also gets to ``know'' the events
in other threads that transitively precede it via the event $a$.

We also maintain, for each lock $\ell$, a vector clock $\P_\ell$ 
that remembers the WCP-predecessor time $P_r$ of the last 
$\rel{\ell}$ event $r$ seen until then (Line 9).
$\P_\ell$ is needed so that the clock $\P_t$ can be maintained correctly.
Consider the case when an $\acq{\ell}$ event $a$ by thread $t$ is processed. 
Now, an event $b$ such that $b \lwcp a$ but not $b \lwcp prev(a)$
($prev$ refers to the previous event in the thread) has to be such that $b \lwcp r$,
 where $r$ is the last $\rel{\ell}$ event before $a$, in some thread other than $t(a)$. 
This is because WCP edges entering an acquire
event from another thread are due to WCP-HB composition (Rule (c) of WCP).
Now $C_b$ is already known to $P_r$ since $b \lwcp r$. 
So the event $a$ can obtain time of events such as $b$ through clock $\P_\ell$. 
This is achieved on Line~2 of the algorithm.

We shall come back to how $\P_t$ is updated on other events
(release/read/write) as it deals with other components of the state which are
described below. For now, note that $\P_t$ corresponds to the relation $\lwcp$
and $\C_t$ corresponds to $\wcp = (\lwcp \cup \toe)$. Therefore $\C_t$ is
simply obtained by incorporating thread order information obtained from the
local clock $\cN_t$ into $\P_t$ as: $\C_t = \P_t[t := \cN_t]$. Since $\C_t$ can be
derived this way from the components $\P_t$ and $\cN_t$, we choose not to
feature how $\C_t$ is updated in the algorithm. 
%We will use $\C_t$ as a
%shorthand for $\P_t[t := \cN_t]$.

\paragraph{HB clocks $\H_t,\H_\ell$ :} For events $a,b$ if $a \lwcp b$, then $P_b$
should not only receive $a$'s time $C_a$, but also the time $C_c$ of
every event $c$ such that $c \hb a$ (Rule (c)). To account for this, we
maintain a clock $\H_t$ for each thread $t$ in the state. 
% An event $e$ is
% assigned a \emph{HB time} $H_e$ which is the value of $\H_{t(e)}$ right after
% $e$ was performed. 
Right after an event $e$ is performed, the clock $\H_{t(e)}$ holds the value of
the \emph{HB time} $H_e$, defined earlier.
% As stated earlier, $H_e$ represents the join of all the times
% $C_{e'}$ of the events $e'$ such that $e' \hb e$. So if $e' \hb e$ then we are
% going to have $C_{e'} \cle H_e$. 
So going back to $a \lwcp b$, instead of
passing $C_a$ to $P_b$, we pass $H_a$ so that the times of all events $c$ (including $a$)
with $c \hb a$, is received by $b$. 
Again, the clock $\H_t(t)$ is made to refer to $\cN_t$ so that we do not have to specify when
$\H_t(t)$ has to be changed.

We also maintain a vector clock $\H_\ell$ for each lock $\ell$ that stores the
HB time of the last $\rel{\ell}$ event seen till then (Line 9).
Consider, once again, the case when an $\acq{\ell}$ event $a$ is processed. 
If there is any event $b$ where $b \hb a$ but not $b \hb prev(a)$,
then $b \hb r$ where $r$ is the last $\rel{\ell}$ event before $a$.
To ensure that $b$'s time $C_b$ reaches $a$ we pass $H_r$ to $a$ 
using $\H_\ell$. $\H_\ell$ contains the
time of all events $b$ such that $b \hb r$. 
This is achieved in Line 1 of the algorithm.

% Note that for any
% such event $c$ it need not be that $c \wcp a$ i.e., passing $C_a$ to $b$ does
% not guarantee that $b$ receives $C_c$.

% This is due to the fact that successive updates to $\C_t$ are only going
% increase it.

In order to motivate the remaining components of the state,
we look at an example trace in Figure~\ref{fig:algoex}. 
We use the event $\acrl{y}$ as a
short hand for the events $\acq{y}\rel{y}$ performed in succession. 
This way two $\acrl{y}$s are HB related. 
An edge between two events in the trace indicates
a WCP order between them that can be deduced using Rules (a) or (b); 
the other ordering edges are omitted for clarity.
% For the sake of clarity the WCP edges
% produced as a result of applying Rules (a) or (b) are indicated in the trace.
As before, we will use $e_i$ to denote the event on Line~$i$.

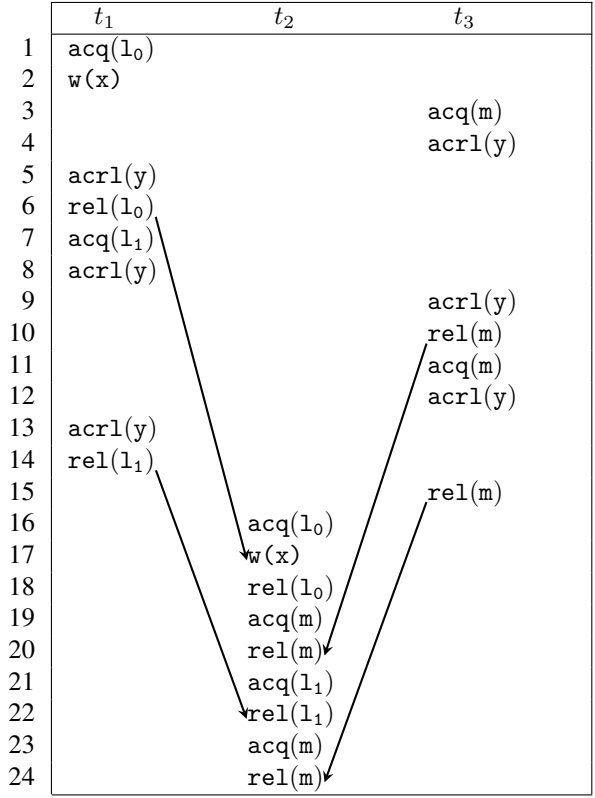
\begin{figure}[ht]
\centering
\begin{tabularx}{0.45\textwidth}{r|XXX|}
\cline{2-4}
 & \quad $t_1$ & \quad $t_2$ & \quad $t_3$\\
\cline{2-4}
1 & $\acq{l_0}$ & & \\
2 & \W{x} & & \\
3 & & & $\acq{m}$ \\
4 & & & $\acrl{y}$ \\
5 & $\acrl{y}$ & & \\
6 & $\rel{l_0}$\tikzmark{e6} & & \\
7 & $\acq{l_1}$ & & \\
8 & $\acrl{y}$ & & \\
9 & & & $\acrl{y}$ \\
10 & & & \tikzmark{e10}$\rel{m}$ \\
11 & & & $\acq{m}$ \\
12 & & & $\acrl{y}$\\
13 & $\acrl{y}$ & & \\
14 & $\rel{l_1}$\tikzmark{e14}& & \\
%15 & & \W{z} & \\
15 & & & \tikzmark{215}$\rel{m}$ \\
16 & & $\acq{l_0} $& \\
17 & & \tikzmark{e17}\W{x} & \\
18 & & $\rel{l_0}$ & \\
19 & & $\acq{m}$ & \\
20 & & $\rel{m}$\tikzmark{e20} &  \\
21 & & $\acq{l_1}$ & \\
22 & & \tikzmark{e22}$\rel{l_1}$ & \\
23 & & $\acq{m}$ & \\
24 & & $\rel{m}$\tikzmark{e24} & \\
\cline{2-4}
\end{tabularx}
\caption{Example trace to motivate the Algorithm~\ref{algo:update}}\label{fig:algoex}
\begin{tikzpicture}[overlay, remember picture, >=stealth, shorten >= 1pt, shorten <= 1pt, thick]
    \draw [->] ({pic cs:e6}) to ({pic cs:e17});
    \draw [->] ({pic cs:e10}) to ({pic cs:e20});
    \draw [->] ({pic cs:e14}) to ({pic cs:e22});
    \draw [->] ({pic cs:215}) to ({pic cs:e24});
    %\draw [->] %([yshift=.75pt]{pic cs:a}) -- ({pic cs:c});
\end{tikzpicture}
\end{figure}

\paragraph{Release Times $\L_{\ell,x}^r,\L_{\ell,x}^w$ :}
Consider the edge $e_6 \lwcp e_{17}$ derived from Rule (a).
Here $e_{17}$ should receive the HB time of $e_6$,
 and more generally,
the HB time of any other $\rel{\ell}$ event $r$ on lock $\ell$ 
enclosing $e_{17}$ ($e_{17}\in\ell$)
such that the critical section $\cs(r)$
contains an event conflicting with $e_{17}$.
For this purpose, we maintain, for every lock $\ell$ and variable $x$,
a vector clock $\L_{\ell,x}^r$ that records the
join of the HB times of all $\rel{\ell}$ events (seen so far)
whose critical sections contain a {\mytt r}$(x)$ event.
Similarly, for each lock $\ell$ and variable $x$, $\L_{\ell,x}^w$ 
maintains the join of the HB times of all the $\rel{\ell}$ events (seen so far)
whose critical sections contain a {\mytt w}$(x)$ event. 
Lines 7 and 8 of Algorithm~\ref{algo:update} maintain these invariants.
Hence, when $e_{17}$ is processed it can look up the time of
$e_6$ using $\L^w_{\tt{l_0},x}$.
This is achieved in Line 11 of the algorithm.
Similarly when a {\mytt w}$(x)$ event is encountered within a critical section of lock $\ell$, the times of the relevant releases (those that contain either read or
write of $x$) can be accessed using $\L^r_{\ell,x},\L^w_{\ell,x}$. 
This is achieved in Line 12 of Algorithm~\ref{algo:update}. 

\paragraph{FIFO Queues $Acq_\ell(t),Rel_\ell(t)$ :}Consider the WCP edge $e_{10}
\lwcp e_{20}$ introduced because of Rule (b) --- the critical sections of
$e_{10}$  and $e_{20}$ contain WCP ordered events (because $e_3 \tho e_4 \hb
e_5 \tho e_6 \lwcp e_{17} \tho e_{20}$ giving  $e_3 \lwcp e_{20}$ by Rule (c)).
Note that two events in two critical sections are  WCP ordered iff the acquire
of the first is WCP ordered to the release of the second (this follows from
thread order and Rule (c)). When processing $e_{20}$, if we somehow knew that
$e_3 \lwcp e_{20}$ then we would need the HB time of $e_{10} = \mtc(e_3)$ to
be communicated to $e_{20}$ because $e_{10} \lwcp e_{20}$ by Rule (b).
Therefore the appropriate times of the critical section $e_3,e_{10}$ 
(and possibly of other critical sections on the same lock performed before $e_{20}$)
need to be stored in the state so they can be used for future reference by
events such as $e_{20}$. Note that, once $e_{20}$ receives the time of
$e_{10}$, no future critical section over lock \code{m}, if any, in thread $t_2$ 
(for example, $e_{24}$ in this case) is required explicitly receive the time of
$e_3$, because a previous release in the same thread ($e_{20}$) has
already received the time of the appropriate release event ($e_{10}$).
Thus, as far as $t_2$ is concerned, once the critical section 
$e_3,e_{10}$ is processed at $e_{20}$, it can be discarded. 
% Therefore, we need to maintain a \emph{queue}
% of critical sections for a particular lock (in this case \code{m}) that are to
% be received by the release events of the same lock performed by a particular
% thread (in this case $t_2$). 
Therefore, for each lock $\ell$ and thread $t$ (in this case \code{m} and $t_2$),
we need to accumulate, in a \emph{queue}, appropriates times for 
critical sections on lock $\ell$, to be later used when
a $\rel{\ell}$ events of thread $t$ is encountered. 
With this in mind, we maintain, in the state of our algorithm,
two FIFO queues $Acq_\ell(t),Rel_\ell(t)$ for every lock $\ell$ and every thread $t$.
These queues will store the times of $\acq{\ell}/\rel{\ell}$ events (in
chronological order), performed by other threads $t'$ ($\neq t$)
%to be later used while processing some $\rel{\ell}$ event in thread $t$ 
(Lines 3 \& 10). 
When processing a $\rel{\ell}$ event performed by thread $t$, 
the algorithm looks up the times of the critical section in the
front of the queue and removes it from the queue
if Rule (b) is applicable, and further, updates its own time 
(Lines 4-6 in Algorithm~\ref{algo:update}).

%We will now see how the state is initialized. 
\paragraph{Initialization :}
The vector clocks $\P_t,\P_\ell,\H_\ell,\L^r_{\ell,x},\L^w_{\ell,x}$ for any thread $t$,
lock $\ell$ and variable $x$ are initialized to $\bot$. For every thread $t$ ,
the local clock $\cN_t$ is initialized to 1 and  the vector clock $\H_t$ is
initialized to $\bot[t := \cN_t]$. Each of the queues $Acq_\ell(t),Rel_\ell(t)$
is empty to begin with.

%!TEX root = main.tex

\SetKwProg{myalg}{procedure}{}{}

\begin{algorithm}
\SetKwFunction{facq}{acquire}
\SetKwFunction{frel}{release}
\SetKwFunction{fork}{fork}
\SetKwFunction{join}{join}
\SetKwFunction{read}{read}
\SetKwFunction{wr}{write}

\myalg{\facq{$t$, $\ell$}}{ %
	\nl $\H_t$ := $\H_t \mx \H_\ell$;\\
	\nl $\P_t$ := $\P_t \mx \P_{\ell}$;\\
	\nl \lForEach{$t' \neq t$}{$Acq_\ell(t').\texttt{Enque}(\C_t)$ } 
}

\myalg{\frel{$t$, $\ell$, $R$, $W$}}{
	%\nl \Repeat{$\P_t$ does not change}{
		\nl \While{ $Acq_{\ell}(t).\textup{\texttt{Front()}} \cle \C_t$ }{ %$\sz{Acq_{\ell}(t)} > 0$ and
			\nl $Acq_\ell(t).\texttt{Deque()} $;\\
			\nl $\P_t := \P_t \mx Rel_\ell(t).\texttt{Deque()}$;\\
			% \nl {\it let $i_{max}$ {be the max} $i$ such that
			%a) $\lang C_a, C_r, H_r \rang = S_\ell(t',t)[i]$\\
			 % $Acq_\ell(t',t)[i] \cle \C_t$}\\
			% \nl $\P_t$ := $\P_t \mx Rel_\ell(t',t)[i_{max}]$;\\
			% \nl remove first $i_{max}{+}1$ entries of $Acq_\ell(t',t)$ and $Rel_\ell(t',t)$;
		}
	%}
	\nl \lForEach{$x \in R$}{ $\L_{\ell,x}^r := \L_{\ell,x}^r \mx \H_t$ }
	\nl \lForEach{$x \in W$}{ $\L_{\ell,x}^w := \L_{\ell,x}^w \mx \H_t$ }
	\nl $\H_{\ell} := \H_t$; $\P_{\ell} := \P_t$;\\
	\nl \lForEach{$t' \neq t$}{
		%\nl $\lang C_a,C_r,H_r \rang := S_\ell(t,t')[\text{-}1]$;\\
		$Rel_\ell(t').\texttt{Enque}(\H_t)$
	}
}
\myalg{\read{$t$, $x$, $L$}}{
	\nl $\P_t$ := $\P_t \mx_{\ell \in L} \L_{\ell,x}^w$
}
\myalg{\wr{$t$, $x$, $L$}}{
	\nl $\P_t$ := $\P_t \mx_{\ell \in L} (\L_{\ell,x}^r \mx \L_{\ell,x}^w)$
}
% \myalg{\fork{$t_p$, $t_c$}}{
% 	\nl $\H_{t_c} := \C_{t_p} \mx \H_{t_p}$; $\P_{t_c} := \C_{t_p} \mx \H_{t_p}$;\\
% 	\nl $N_{t_p}$++;
% }
% \myalg{\join{$t_p$, $t_c$}}{
% 	\nl $\H_{t_p} := \H_{t_p} \mx \H_{t_c} \mx \C_{t_c}$;\\
% 	\nl $\P_{t_p} := \P_{t_p} \mx \H_{t_c} \mx \C_{t_c}$;
% }
\caption{\textit{Updating vector clocks on different events}}
\label{algo:update}
\end{algorithm}

Next we state the correctness of the algorithm which states the correspondence
between the ordering of timestamps assigned to events ($C_e$ for event $e$)
and the WCP ordering. The proof the theorem is provided in
\iftoggle{techreport}{Appendix~\ref{app:algo}}{\cite{techreport}}.

% \begin{thm}[Correctness of Algorithm]\label{thm:algorithm}
% Given any trace $\sig$, for any two events $a,b$
% where $a \trs b$ we have
%  $a \wcps b \iff C_{a} \cle C_{b}$
% \end{thm}
\begin{thm}[Correctness of Algorithm~\ref{algo:update}]\label{thm:algorithm}
For a trace $\sig$ and events $a,b$
with $a \trs b$, we have 
 $a \wcps b \iff C_{a} \cle C_{b}$
\end{thm}

Theorem~\ref{thm:algorithm} tells us that two events 
$a$ and $b$ where $a \tr b$ are in
WCP-race exactly when $C_a$ and $C_b$ are incomparable. 
This yields an algorithm for checking all the WCP-races for the trace. 
For each variable $x$, we maintain vector
clocks $\Rc_x$ and $\Wc_x$ that
record the join of the $C_e$ times of all the 
read and write events $e$ on the variable $x$
that have been seen for the prefix of the trace that is processed. 
On encountering a read event $e = \mathtt{r}(x)$, we check if
$\Wc_x \cle C_e$ to confirm
that all earlier events conflicting with $e$ are
indeed ordered before $e$.
If this check fails, then $e$ is in WCP-race with some earlier
conflicting write event, and thus we can declare a warning. 
Similarly for a write event $e = \mathtt{w}(x)$,
we check if $\Rc_x \mx \Wc_x \cle C_e$ holds
and declare a warning otherwise. 
Note that our soundness
theorem only guarantees that the first race pair is an actual race. 
But in practice we have observed that subsequent pairs that are 
in WCP-race also happen to be in race.  
Note that this methodology only gives us 
the second component $e_2$ of a pair $(e_1, e_2)$ of events in race. 
In order to determine the first part,
we would have to go over the trace once more and 
individually compare the WCP times of the
events against those conflicting events appearing later that were flagged to
be in race in the initial analysis.

\subsection{Linear Running Time}

In order to analyze the running time of Algorithm~\ref{algo:update}, we fix the
following parameters for a given trace $\sig$. 
Let $\Nt$ be the number of threads, $\Nl$ be the number of 
locks used in $\sig$, and $\Ne$ be the total
number of events in $\sig$. 
We state the running time in Theorem~\ref{thm:runningtime}
and provide the proof in  \iftoggle{techreport}{Appendix~\ref{app:time}}{\cite{techreport}}.
Note that we assume arithmetic operations take constant time.

\begin{thm}\label{thm:runningtime} Given a trace $\sig$ with parameters $\Ne,\Nt$ and $\Nl$ as
defined above, the total running time of the WCP vector-clock algorithm over
$\sig$ is proportional to $\Ne{\cdot}(\Nl + \Nt^2)$. \end{thm}

Note that for most applications the parameter $\Nt$ is usually
small ($< 25$) and parameter $\Nl$ is
at most a few thousand. 
Typically, the bottleneck for any online race-detection technique
is the length of the trace $\Ne$ which can be of the order of hundreds of
millions or even billions ($10^8$,$10^9$) especially for industrial scenarios.
Our proposed algorithm is \textit{\textbf{linear}} in the size of the trace
$\Ne$ (unlike CP/\rvpredict{}) and therefore truly scales to large traces
without having to rely on any windowing strategy that would restrict
the scope of the races detected. 
Our experimental evaluation (Section~\ref{sec:experiments}) supports
this claim by showing that the algorithm scales well in practice.

\subsection{Lower Bounds}\label{sec:cp-hard}

  Our algorithm is clearly optimal
in terms of running time (in terms of the length of the trace), since one has
to spend linear time just looking at the entire trace. 
However, it can take up linear space in the worst case owing to the queues that hold the
times of the previous acquires/releases. 
In this section, we state lower bounds
result that states that any algorithm that implements WCP in linear time takes
linear space. Hence our algorithm is also optimal in terms of the space
requirements.

The goal is to show that one needs linear space to recognize WCP. But we want
to do this when the number of threads $\Nt$ is constant, and the number of
variables $\Nv$ and locks $\Nl$ is sublinear in the size of the trace
$(O(\frac{n}{log{n}}))$. If the number of variables and locks is linear then
it is easy to get a linear space lower bound, and
even HB takes linear space in that case. 
The following theorem is proved in
\iftoggle{techreport}{Appendix~\ref{app:lowerbounds}}{\cite{techreport}}.

\begin{thm}\label{thm:lb1}
Any algorithm that implements WCP by doing a single pass over the trace
takes $\Omega(n)$ space.
\end{thm}

% The formal details of the argument are provided in \iftoggle{techreport}{Appendix~\ref{app:lowerbounds}}{\ucomment{arxiv}}.

In Theorem~\ref{thm:lb2}, we prove a lower bound result that applies to any 
algorithm (not just single pass), with details in 
\iftoggle{techreport}{Appendix~\ref{app:lowerbounds}}{\cite{techreport}}.

\begin{thm}\label{thm:lb2}
For any algorithm that computes the WCP relation in time $T(n)$ and space $S(n)$
it is the case that $T(n)S(n) \in \Omega(n^2)$
\end{thm}

Therefore, our WCP algorithm is optimal in terms of
time/space trade-off as well.

%!TEX root = main.tex

\section{Experimental Evaluation}
\label{sec:experiments}

We implemented the vector clock algorithm (Algorithm~\ref{algo:update})
in our tool \tool~(\textsc{Ra}ce \textsc{P}re\textsc{d}iction) written in Java, available at~\cite{rapid}.
\tool~uses the logging functionality of the tool \rvpredict~\cite{rvpredict},
to generate program traces which can be used for race detection.
% Since our sound relation WCP does not use any information about the \emph{branch} 
% events in \rvpredict, \tool~only considers events corresponding to 
% read/write on variables, lock acquires and releases, and thread forks and joins.
\tool~only considers events corresponding to read/write to memory locations,
lock acquire/release, and thread fork/join events, and ignores
other events (such as \emph{branch} events) generated by \rvpredict.
\tool~also implements a vector algorithm for detecting HB races.

Algorithm~\ref{algo:update} runs in linear time.
Also, since the sizes of data structures involved
in the algorithm do not grow very fast, 
even for very large benchmarks,
the memory requirement for the algorithm
did not seem to be a huge bottleneck.
Hence, we did not have to split our analysis into small windows, to handle large traces in \tool.
This is in stark contrast to most predictive dynamic race detection techniques
\cite{rv2014,cp2012}.
% In fact, our experimental evaluation (Section~\ref{sec:experiements}) suggests that quite often
% data races occur between events that are separated 
% by a large number of events between them, and any technique that
% disregards the possibility of occurrence of such races
% can potentially hamper its race detection capability.

%\ucomment{Talk about optimizations ?. Specially stack ?}

% We evaluate the performance of our algorithm for detecting 
% WCP races implemented in \tool.
Our experiments were conducted on an 8-core 2.6GHz 46-bit Intel Xeon(R) Linux machine,
with HotSpot 1.8.0 64-Bit Server as the JVM and 50 GB heap space.
We set a time limit of 4 hours for evaluating each of the techniques on each of the benchmarks.
Our results are summarized in Table~\ref{tab:exp}.
%!TEX root = main.tex

% \usepackage{multirow}

\begin{table*}[t!]
\centering
\scalebox{0.8}{
\begin{adjustbox}{center}

\begin{tabular}{|c|c|c|c|c|c|c|c|c|c|c|c|c|c|c|}
\hline
1 & 2 & 3 & 4 & 5 & 6 & 7 & 8 & 9 & 10 & 11 & 12 & 13 & 14 & 15\\ \hline
\rowcolor[HTML]{EFEFEF} 
\cellcolor[HTML]{EFEFEF} 
& \cellcolor[HTML]{EFEFEF} 
& \cellcolor[HTML]{EFEFEF} 
& \cellcolor[HTML]{EFEFEF} 
& \cellcolor[HTML]{EFEFEF} 
& \multicolumn{5}{c|}{\cellcolor[HTML]{EFEFEF}\textbf{\#Races}} 
& \cellcolor[HTML]{EFEFEF} \textbf{WCP}
& \multicolumn{4}{c|}{\cellcolor[HTML]{EFEFEF}\textbf{Time}} \\ %\cline{6-10} \cline{12-15} 
\rowcolor[HTML]{EFEFEF} 
\cellcolor[HTML]{EFEFEF} 
& \cellcolor[HTML]{EFEFEF} 
& \cellcolor[HTML]{EFEFEF} 
& \cellcolor[HTML]{EFEFEF} 
& \cellcolor[HTML]{EFEFEF} 
& \cellcolor[HTML]{EFEFEF} 
& \cellcolor[HTML]{EFEFEF} 
& \multicolumn{3}{c|}{\cellcolor[HTML]{EFEFEF}\textbf{RV}} 
& \cellcolor[HTML]{EFEFEF} \textbf{Queue} 

& \cellcolor[HTML]{EFEFEF} 
& \cellcolor[HTML]{EFEFEF} 
& \multicolumn{2}{c|}{\cellcolor[HTML]{EFEFEF}\textbf{RV}} \\ %\cline{9-11} \cline{14-15} 
\rowcolor[HTML]{EFEFEF} 
\multirow{-3}{*}{\cellcolor[HTML]{EFEFEF}\textbf{Program}} 
& \multirow{-3}{*}{\cellcolor[HTML]{EFEFEF}\textbf{LOC}} 
& \multirow{-3}{*}{\cellcolor[HTML]{EFEFEF}\textbf{\#Events}} 
& \multirow{-3}{*}{\cellcolor[HTML]{EFEFEF}\textbf{\#Thrd}} 
& \multirow{-3}{*}{\cellcolor[HTML]{EFEFEF}\textbf{\#Locks}} 
% & \multirow{-2}{*}{\cellcolor[HTML]{EFEFEF}\textbf{Length (\%)}} 
& \multirow{-2}{*}{\cellcolor[HTML]{EFEFEF}\textbf{WCP}} 
& \multirow{-2}{*}{\cellcolor[HTML]{EFEFEF}\textbf{HB}} 
& \textbf{\begin{tabular}[c]{@{}c@{}}w=1K\\ s=60s\end{tabular}} 
& \textbf{\begin{tabular}[c]{@{}c@{}}w=10K\\ s=240s\end{tabular}} 
& \textbf{Max} 
& \textbf{\begin{tabular}[c]{@{}c@{}}Length\\ (\%)\end{tabular}} 
& \multirow{-2}{*}{\cellcolor[HTML]{EFEFEF}\textbf{WCP}} 
& \multirow{-2}{*}{\cellcolor[HTML]{EFEFEF}\textbf{HB}} 
& \textbf{\begin{tabular}[c]{@{}c@{}}w=1K\\ s=60s\end{tabular}} 
& \textbf{\begin{tabular}[c]{@{}c@{}}w=10K\\ s=240s\end{tabular}} \\ \hline
account 		& 87 	& 130 	& 4		& 3 	& \cellcolor[HTML]{EFEFEF}4             & 4     & 4     & 4     & 4     & 0.0   & 0.2s & 0.3s & 1s & 1s \\
airline 		& 83 	& 128 	& 2 	& 0 	& \cellcolor[HTML]{EFEFEF}4             & 4     & 4     & 4     & 4     & 0.0   & 0.2s & 0.2s & 0.8s & 2s \\
array 			& 36 	& 47 	& 3 	& 2 	& \cellcolor[HTML]{EFEFEF}0             & 0     & 0     & 0     & 0     & 4.3   & 0.2s & 0.2s & 1.1s & 0.8s \\
boundedbuffer	& 334 	& 333 	& 2 	& 2 	& \cellcolor[HTML]{EFEFEF}2             & 2     & 2     & 2     & 2     & 0.0   & 0.3s & 0.2s & 1s & 0.8s \\
bubblesort 		& 274 	& 4K 	& 10 	& 2 	& \cellcolor[HTML]{EFEFEF}6             & 6     & 6     & 0     & 6     & 2.4   & 0.7s & 0.5s & 3.6s & 7m3s \\
bufwriter 		& 199	& 11.7M & 6 	& 1 	& \cellcolor[HTML]{EFEFEF}2             & 2     & 2     & 2     & 2     & 10    & 47s & 22.4s & 4.1s & 4.5s \\
critical		& 63 	& 55 	& 4 	& 0 	& \cellcolor[HTML]{EFEFEF}8             & 8     & 8     & 8     & 8     & 0.0   & 0.2s & 0.2s & 1.7s & 0.9s \\
mergesort 		& 298 	& 3K 	& 5 	& 3 	& \cellcolor[HTML]{EFEFEF}3             & 3     & 1     & 2     & 2     & 1.3   & 0.4s & 0.4s & 1.1s & 1.4s \\
pingpong		& 124 	& 146 	& 4 	& 0 	& \cellcolor[HTML]{EFEFEF}7             & 7     & 7     & 7     & 7     & 0.0   & 0.5s & 0.3s & 1.2s & 1.3s \\ \hline
moldyn			& 2.9K 	& 164K 	& 3 	& 2		& \cellcolor[HTML]{EFEFEF}44            & 44    & 2     & 2     & 2     & 0.0   & 7.1s & 2.4s & 1.4s & 17.4s \\
montecarlo 		& 2.9K	& 7.2M 	& 3 	& 3		& \cellcolor[HTML]{EFEFEF}5             & 5     & 1     & 1     & 1     & 0.0   & 23.4s & 16.2s & 7.1s & 5.7s \\
raytracer 		& 2.9K 	& 16K 	& 3 	& 8 	& \cellcolor[HTML]{EFEFEF}3             & 3     & 2     & 3     & 3     & 0.0   & 2.4s & 1s & 1s & 14.7s \\ \hline
derby 			& 302K 	& 1.3M 	& 4 	& 1112	& \cellcolor[HTML]{EFEFEF}23            & 23    & 11    & -     & 14    & 0.6   & 16.2s & 7s & 31.2s & TO \\
eclipse 		& 560K 	& 87M 	& 14 	& 8263 	& \cellcolor[HTML]{EFEFEF}\textbf{66}   & 64    & 5     & 0     & 8     & 0.4   & 6m51s & 4m18s & 26.2s & 15m10s \\
ftpserver		& 32K 	& 49K 	& 11 	& 304 	& \cellcolor[HTML]{EFEFEF}36            & 36    & 10    & 12    & 12    & 2.2   & 5.7s & 2.1s & 3.8s & 3m \\
jigsaw 			& 101K 	& 3M 	& 13 	& 280 	& \cellcolor[HTML]{EFEFEF}\textbf{14}   & 11    & 6     & 6     & 6     & 0.0   & 18s & 11.8s & 2.8s & 14.7s \\
lusearch 		& 410K 	& 216M 	& 7 	& 118 	& \cellcolor[HTML]{EFEFEF}160           & 160   & 0     & 0     & 0     & 0.0   & 10m13s & 6m48s & 57.3s & 46.7s \\
xalan 			& 180K 	& 122M 	& 6 	& 2494 	& \cellcolor[HTML]{EFEFEF}\textbf{18}   & 15    & 7     & 8     & 8     & 0.1   & 7m22s & 4m46s & 43.1s & 7m11s \\
\hline
\end{tabular}

\end{adjustbox}
}
\caption{
Experimental results : 
Columns 1-2 describe the benchmarks (name and lines of source code respectively).
Columns 3, 4 and 5 denote the number of events, threads and locks in the generated trace.
Columns 6, 7 denote the number of distinct race pairs detected by 
\tool~by running WCP and HB analysis respectively.
Columns 8 and 9 denote the number of races detected by \rvpredict~when the window sizes
are respectively 1K and 10K and the solver timeouts are respectively 60 and 240 seconds.
For programs for which WCP detects more races than HB, the corresponding entries
in Column 6 are boldfaced.
Column 10 represents the maximum number of races detected by \rvpredict~with
all combinations of windows sizes (1K, 2K, 5K, 10K) and solver timeouts (60s, 120s, 240s).
Column 11 reports the maximum value of the sum (over locks $\ell$ and threads $t$) of 
lengths of the queues $Acq_\ell(t)$ and $Rel_\ell(t)$
(Algorithm~\ref{algo:update}) attained at any point while performing WCP analysis on the generated trace,
as a fraction of \#events.
Columns 12 and 13 respectively denote the time taken by \tool~for WCP and HB analysis.
Columns 14 and 15 respectively denote the time taken by \rvpredict~for window sizes 1K and 10K,
and solver timeouts 60s and 240s respectively.
A `-' in Column 9 and `TO' in Column 15 represents timeout which we set to be 4 hours.
}
\label{tab:exp}
\end{table*}

% Note that the soundness guarantee provided by WCP
% applies only to the first conflicting pair of events. 
% Hence, we manually inspected each of the races reported by our algorithm.
% Indeed each of the races reported by the WCP vector clock
% algorithm  (Column 6) turn out to be real races and no deadlocks
% were detected.

We compare the performance of \tool~only
against sound techniques for race detection.
For each of the techniques compared, as in \cite{rv2014},  
we attempt to analyze the following two characteristics of our algorithm:
\begin{enumerate}
	\item \emph{Race detection capability}, measured by the 
	number of distinct race pairs detected.
	A WCP (HB) race pair is an unordered tuple of \emph{program locations} corresponding to
	some pair of events in the trace that are unordered by the partial order WCP (HB).
	We compare the race detection capability of \tool's WCP 
	vector clock algorithm with \rvpredict~(version 1.8.2), which,  
	in theory, detects at least as many races
	as any sound dynamic race detection technique.
	\item \emph{Scalability}, measured by the time taken to analyze the entire trace.
	We compare our analysis time against HB vector clock algorithm
	for race detection (also implemented in \tool), since HB
	is the simplest sound technique, and admits a fast linear time algorithm.
\end{enumerate}
\vspace{-0.1in}
As stated earlier, WCP detects all the races detected using CP.
Further, it is not clear if an algorithm based on CP relation
can scale without windowing strategy. 
Therefore, we omit any comparison with CP~\cite{cp2012}.
% Also, it is unlikely that CP will admit a linear time algorithm
% because past efforts for designing such an algorithm have failed.
% This is also indicated in Example~\ref{}.
% We leave the proof for a non-linear lower bound on the running time
% of any algorithm for detecting CP races for future work.

We run all the three techniques (HB, WCP, \rvpredict) on the
same set of traces.
This was possible because both \rvpredict~and \tool~can analyze a 
logged trace produced by \rvpredict's logging feature.
We have tried to ensure that the comparison is fair;
we implement the linear time vector clock algorithm for detecting 
HB~\cite{Mattern1988}, and do not restrict the HB analysis to small windows,
unlike in~\cite{rv2014} and~\cite{cp2012}.

\rvpredict~supports tuning of parameters like window sizes and 
timeout for its backend SMT solver.
%to ensure that the entire analysis completes in a reasonable amount of time.
The tight interplay between window sizes and 
solver timeout in \rvpredict~makes it difficult to estimate the best combination
of these parameters. 
Small windows result in a low number
of reported races because every occurrence of most of the races
occur across multiple windows. 
On the other hand, a large window
implies that the logical formulae generated in \rvpredict~are too large
to be solved for the SMT solver, within the timeout, as a result of which,
most of the windows do not report any races.
\begin{figure}[t]
\includegraphics[width=0.45\textwidth]{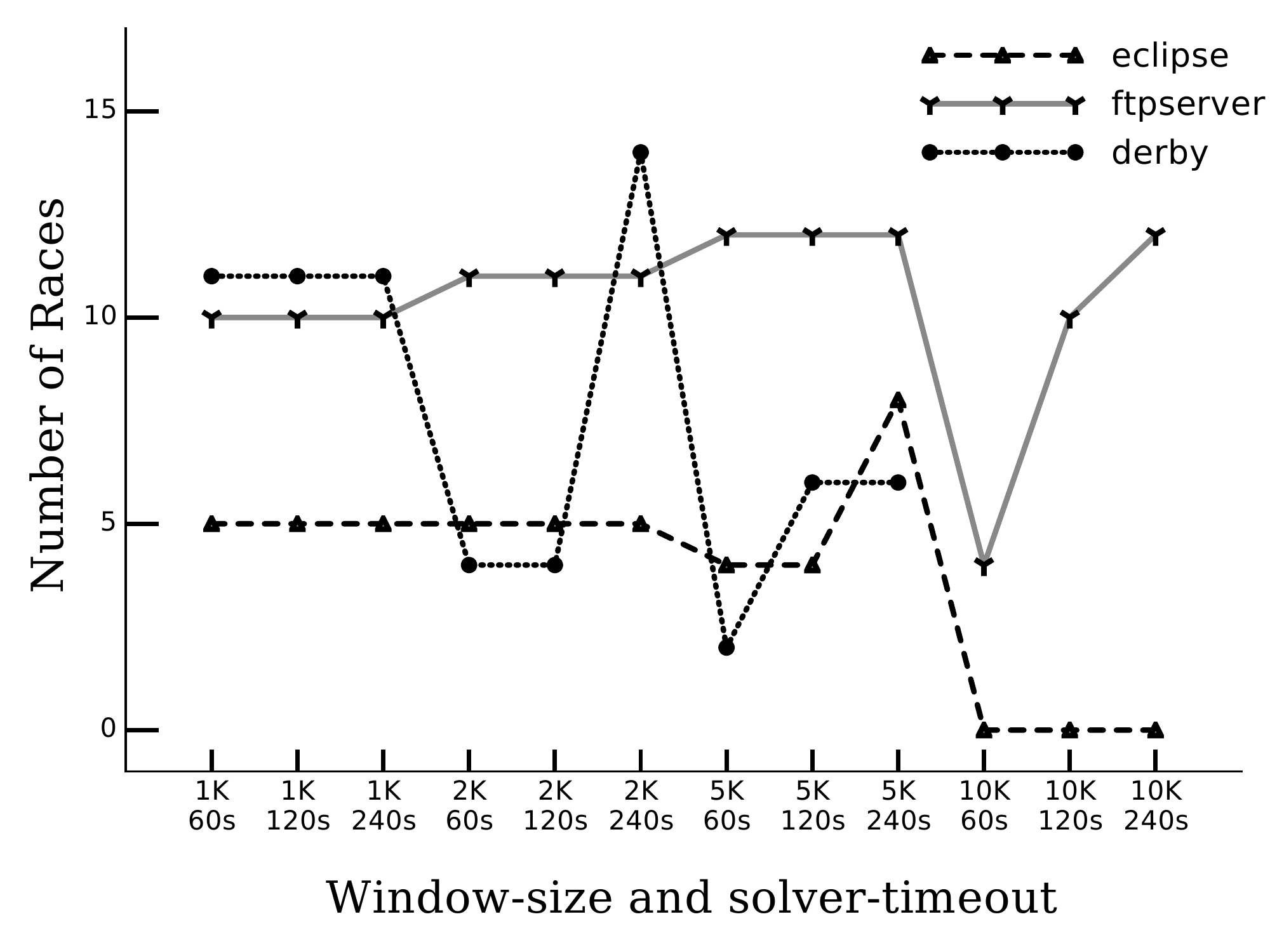}
\caption{Number of races detected by \rvpredict~for different values of window size and solver timeout parameters}
\label{fig:rvraces}
\end{figure}
In Figure~\ref{fig:rvraces}, we depict how the number of races
vary for different values of solver timeouts and window sizes, for three benchmark examples,
and as can be seen, there is no clear pattern.
We run \rvpredict~on each of the benchmarks with several parameter combinations;
we vary window sizes as (1K, 2K, 5K, 10K) and 
solver timeout values as (60s, 120s, 240s).
An attempt of testing \rvpredict~beyond these parameter values often led to 
very large running times or excessive memory requirements.
In Table~\ref{tab:exp}, we report the observations
only for two of these combinations.

\subsection{Benchmarks}
Our evaluation benchmarks (Column 1) have primarily been derived from \cite{rv2014}.
The benchmarks are designed for a comprehensive performance evaluation :
the lines of code range from 60 to 0.5M, and the number of
events vary from an order of 10 to 200M.
The first set of small-sized benchmarks (account to pingpong) and
are originally derived from IBM Contest benchmark suite~\cite{Farchi2003}.
The second set of medium sized benchmarks are derived from the Java Grande Forum
benchmark suite~\cite{JGF2001}.
The third set of benchmarks come from large real world software -
Apache FTPServer, W3C Jigsaw web server, Apache Derby, and some applications
from the DaCaPo benchmark suite (version 9.12) ~\cite{DaCapo2006}.

\subsection{Scalability}
Columns 12-15 report the times taken by WCP, HB and \rvpredict.
WCP analysis times are comparable to HB analysis for all the examples.

For the small set of examples, all the three
techniques finish their analysis in a reasonably small time, roughly
proportional to the size of the traces.
However, for large examples, both WCP and HB outperform the running times
of \rvpredict. 
For `derby', \rvpredict~exceeds the time limit of 4 hours
for large window sizes.
In general, it is difficult to gauge the running
times of \rvpredict~from the length of the trace
and window size, because the actual size
of the internal logical formulation generated by
 \rvpredict~for every window, 
and the running times of its backend SMT
solver crucially depend on how complex these windows are.

We highlight that the worst case linear space complexity of Algorithm~\ref{algo:update}
was not observed in our experiments.
% We emphasize that the worst case linear space complexity of Algorithm~\ref{algo:update}
% does not bottleneck the performance of our algorithm.
In Column 11, we report the maximum value of the total lengths of the FIFO queues
(Section~\ref{sec:algo}),
attained at any point while performing WCP analysis on the traces,
as a percentage of the number of events.
As can be seen, this fraction stays below $3\%$ for almost all examples,
and is $10\%$ for `bufwriter' benchmark.

\subsection{Bug Detection Capability}
Columns 6 and 7 report the number of distinct HB and WCP race pairs.
Columns 8-10 report the number of race pairs detected by \rvpredict,
with different parameters. 
The extra races discovered by WCP and not by HB (boldfaced entries in Column 6)
 were either found to be \rvpredict~races,
or were manually inspected to be valid race pairs.

For the smaller benchmarks, when the number of events is relatively small,
the number of race pairs detected by the three techniques is almost the same, with
WCP and HB detecting the maximum number of races for each of these benchmark examples,
despite the fact that \rvpredict~employs a theoretically
more comprehensive technique than WCP.

For the larger benchmark examples, the number of races detected 
by \rvpredict~are much lower than those predicted by WCP or HB.
For the `derby' benchmark, \rvpredict~could not
finish its analysis within 4 hours.
In fact, for the benchmark `bubblesort', \rvpredict~runs out of
memory (50GB) for a window size of 5K.
This is primarily because of the tight interplay between
the window sizes and timeout parameter for the backend SMT solver.
We conjecture that most of the races that are not reported by
\rvpredict~either occur across windows, or
are missed because the solver did not return with an answer within the specified time limit.
On the other hand, the theoretical guarantee provided
by the techniques used in \rvpredict~\cite{rv2014}
would have ideally resulted in it detecting
possibly more data races than both HB and WCP.
This indicates that a windowing strategy for analyzing large 
traces can potentially result in  significant loss in the bug detection 
capability of any dynamic race detection technique.
In fact, on careful analysis of the predicted races,
we found that both HB and WCP detect races having a \emph{distance} of
millions of events --- the \emph{distance}
 of a race between program locations $(pc_1,pc_2)$
is the minimum separation (in terms of the number of events in the trace) 
between any pair of events $(e_1, e_2)$ 
exhibiting the $(pc_1, pc_2)$ race.
Specifically, both HB and WCP expose more than 25 races in `eclipse' 
having a distance of at least $4.8$ million events, 
with the maximum distance being 53 million.
Clearly, any windowing based analysis will be incapable of
catching these races.
% \todo{complete}

In all the examples, the \emph{set} of races detected by
HB are a subset of the \emph{set} of races detected by WCP.
This is expected as WCP is a weakening of CP (and hence of HB).
For the large benchmarks `eclipse', `jigsaw' and `xalan', 
the number of races detected by 
WCP are more than those detected by HB.
% For all examples, except for `mergesort' and `ftpserver',
% the \emph{set} of races detected by \rvpredict~are a subset of the races reported by
% WCP and HB---\rvpredict~reports one different race in each of 
% `mergesort' and `ftpserver', which are not detected by HB or WCP.
In all examples, the \emph{set} of races detected by \rvpredict~are 
a subset of the races reported by
WCP and HB, except in `mergesort' and `ftpserver', where 
\rvpredict~reports one extra race each.

Note that, our WCP based race detection algorithm does not report
drastically more races than the simpler HB based algorithm.
While this is surprising, given the optimistic comparison of CP and \rvpredict~ versus
HB, as reported in~\cite{cp2012,rv2014}, the apparent disparity can be explained by the fact
that both~\cite{cp2012} and~\cite{rv2014} compare their techniques 
against a windowed implementation of HB based vector clock algorithm,
which can potentially miss HB races between pairs of events (unordered by HB)
that occur far apart in the trace, and thus, possibly missing 
out on program location pairs corresponding to these event pairs.
As stated before, our implementation of HB vector clock algorithm
is not \emph{windowed} and catches these far-away event pairs unordered by HB.
That being said, the few extra races missed by (our implementation of) HB,
but predicted by WCP and/or \rvpredict, are quite subtle.

%!TEX root = main.tex
\section{Related Work}\label{sec:related}

Our work generalizes the causally-precedes (CP) relation
proposed by Smaragdakis et. al~\cite{cp2012}.
WCP is a weaker relation than CP 
(any data race detected by CP will also be detected by WCP)
and can be implemented using a linear time vector clock algorithm.
WCP based race detection can be viewed as belonging to the class
of predictive analysis techniques, similar to CP~\cite{cp2012}, Said et. al~\cite{Said2011}
 \rvpredict~\cite{rv2014}, IPA~\cite{ipa2016}, \cite{SPA2009}, which 
essentially reason about correct reorderings of a given trace 
for estimating concurrency errors in other possible executions.
\rvpredict~\cite{rv2014} and Rosu et. al.,~\cite{maxcausalmodels} both show how maximal and sound causal models can be used to resolve concurrency bugs. 
These methods explore all possible interleavings that can be deduced from
the given trace, but such complete explorations are known to be intractable.
\rvpredict~\cite{rvpredict} has also been extended for analyzing traces 
with missing events~\cite{Huang2016}.
Predictive analysis techniques have also been used for checking
atomicity violations and synchronization errors in tools like
$\textsf{jPredictor}$~\cite{chen-serbanuta-rosu-2008-icse}, \cite{sen2005detecting}  and \textsf{TAME}~\cite{Huang15}.
{GPredict}~\cite{Huang:2015:GGP:2818754.2818856} uses predictive analysis for resolving higher level 
concurrency issues like authentication-before-use in Java.
% Sinha et. al.~\cite{sinha2011predictive} employ trace segmentation
% for predicting serializability violations.
% Gupta et. al.~\cite{tracesets} propose data structures for efficient representation of
% causal models.
Our experimental evaluation suggests that
heavy weight predictive techniques such as SMT based search do not scale well
in practical scenarios, and often forego predictive power for scalability.

Lockset based techniques such as \textsc{Eraser}~\cite{savage1997eraser}, 
which assign sets of locks to program locations and variable accesses
are known to be unsound.
Methods such as random testing~\cite{Sen:2008:RDR:1375581.1375584} and 
static escape analysis~\cite{vonPraun:2001:ORD:504311.504288}
aim to improve
the efficiency of, and reduce the number of false alarms raised
by lockset based analysis.

Other dynamic race detection techniques use Lamport's happens-before (HB) relation.
HB admits a linear time vector clock algorithm~\cite{Mattern1988},
and is adopted by various techniques including~\cite{Pozniansky:2003:EOD:966049.781529,fasttrack,trade}.
The \textsc{DJIT}$^+$~\cite{Pozniansky:2003:EOD:966049.781529} 
algorithm uses the epoch optimization for performance improvement 
in the traditional vector clock algorithm.
This was further enhanced by \textsc{Fasttrack}~\cite{fasttrack}.
Both WCP and CP are weaker relations than HB, and thus, in principle,
detect all the races that any happens-before based race detection algorithm.
 
 Other techniques include
combinations of HB and lockset approach~\cite{elmas2007goldilocks,choi:2003:HDD:781498.781528,threadsanitizer,racetrack},
 statistical techniques~\cite{bond2010pacer,marino2009literace},
 and crowd-sourced inference~\cite{racemob}.
Tools such as \textsc{RoadRunner}~\cite{flanagan2010roadrunner}
and $\mathsf{Sofya}$~\cite{kinneer2007sofya}
 provide frameworks for implementing dynamic analysis tools.
 Techniques such as~\cite{feng1999efficient,cheng1998detecting,raman2012scalable,surendran2016dynamic,Yoga2016} leverage structured parallelism
 to optimize memory overhead for dynamic race detection.

Static race detection techniques~\cite{Naik:2006:ESR:1133255.1134018,racerx,pratikakis11locksmith,voung2007relay,echo,Radoi:2013:PSR:2483760.2483765}
suffer from the undecidability problem,
and raise many false alarms, but still remain popular amongst developers.
% Several tools  like \textsc{RacerX}~\cite{racerx},
% \textsc{Locksmith}~\cite{pratikakis11locksmith}, \textsc{Relay}~\cite{voung2007relay},
% \textsc{ECHO}~\cite{echo} and \textsc{IteRace}~\cite{Radoi:2013:PSR:2483760.2483765} have been developed.
Type systems for detecting concurrency errors~\cite{boyapati2002ownership, flanagan2000type,Abadi:2006:TSL:1119479.1119480},
aim to help programmers write safer programs.

Model checking techniques like~\cite{heisenbugs} and
\cite{Yahav:2001:VSP:373243.360206} aim to exhaustively explore
all possible concurrent executions to detect data races.
However, due to the state explosion problem,
explicit state model checking encounters a huge slowdown.

%!TEX root = main.tex
%\vspace{-2pt}
\section{Conclusion}\label{sec:conclusion} 

In this paper, we presented a sound technique for detecting data races through
a new partial order called Weak-Causally-Precedes (WCP).  
Since WCP is a weakening of CP, it
provably detects more races than CP. We showed how WCP races can be detected in
linear time using a vector clock algorithm and we proved its optimality. We
implemented our techniques in a prototype tool \tool{} which shows promising
results when evaluated on industrial sized benchmarks.

There are several avenues for future work.
These include use of epoch
based optimizations for improving memory requirements of the implementation,
further weakening of the WCP relation while preserving soundness,
and incorporating control flow information for enhanced race detection capability.

%!TEX root = main.tex

\section*{Acknowledgments}
We thank Grigore Rosu and Yilong Li for their help with setting up 
\rvpredict, and Jeff Huang for providing many of the benchmark examples.
We gratefully acknowledge the support of the following grants ---
Dileep Kini was partially supported by NSF TWC 1314485; 
Umang Mathur was partially supported by NSF CSR 1422798; 
and Mahesh Viswanathan was partially supported
by NSF CPS 1329991 and AFOSR FA9950-15-1-0059.
% {\small
% \paragraph{Acknowledgments :}We thank Grigore Rosu and Yilong Li for their help with setting up 
% \rvpredict, and Jeff Huang for providing many of the benchmark examples.
% We gratefully acknowledge the support of the following grants ---
% Dileep Kini was partially supported by NSF TWC 1314485; 
% Umang Mathur was partially supported by NSF CSR 1422798; 
% and Mahesh Viswanathan was partially supported
% by NSF CPS 1329991 and AFOSR FA9950-15-1-0059.}

\bibliographystyle{abbrvnat}  
\bibliography{references}

%!TEX root = main.tex

\clearpage
\appendix

\newtheorem{prop}{Proposition}
\numberwithin{prop}{section}
\newtheorem{lem}{Lemma}
\numberwithin{lem}{section}
\newtheorem{cor}{Corollary}
\numberwithin{cor}{section}

\section{Soundness of WCP}\label{sec:wcpsound}
 Our soundness proof of WCP is along the lines of
that of CP, but crucially differs when it comes down to proving existence of
deadlocks. For CP, when it shown that there is a deadlock because of a
CP-race, the deadlock always involves only two threads.   For WCP, we end up
proving that there could be deadlocks involving more than just two threads. In
doing so we explore deeper structures about deadlocks which are novel
contributions in themselves apart from those mentioned in the paper.

Given a trace $\sig$ with a WCP-race the aim is to obtain a correct reordering
of $\sig$, which exhibits a race or a deadlock. Let
the first WCP-race be between events $e_1$ and $e_2$, where first
means there is no other pair in WCP-race before $e_2$ and no other event
$e_1'$ such that $e_1 \trs e_1' \trs e_2$ , $e_1',e_2$ are in WCP-race. 

Next, among all the correct reorderings of $\sig$ we pick a $\alpha$
which satisfies the following:
\begin{itemize}
\item $e_1,e_2$ is the first race in $\alpha$
\item among traces that satisfy the above $\alpha$ is such that distance between
$e_1,e_2$ is minimal
\item among traces that satisfy the above $\alpha$ is such that it minimizes
the distance from $e_2$ to every acquire that encloses $e_1$ from innermost
to outer acquires.
\end{itemize}

We shall refer to such a trace $\alpha$ as an \emph{extremal trace}. If
$e_1,e_2$ are in HB-race then we can use the proof of HB correctness to get a
correct reordering of the extremal trace (which is also a correct reordering
of the original trace) which exhibits a race. So let us assume $e_1,e_2$ are
not in HB-race, i.e., $e_1 \hb e_2$, and in this case we are going to show a
predictable deadlock. Let $t_1$ be the thread containing $e_1$ and $t_2$ be
the thread containing $e_2$.  We observe the following properties about
$\alpha$ in Lemmas~\ref{lem:between} to \ref{lem:dangle} which are analogous
to the Lemmas 1-5 in \cite{cp2012}, we only state these Lemmas and skip their
proofs as they use exactly the same argument. The weakening of CP to WCP does
not effect any of their reasoning. We provide proofs for the remaining Lemmas
and of course our main theorem.

In all of our arguments we will attempt to reorder/move the events in $\alpha$
(or some trace) which if successful will violate some extremality condition.
Therefore such  a move cannot be possible. Most of the reorderings we make
will involve moving a segment of the events of a thread, say $[a,b]$ to an
earlier point in the trace, say right before $c$ (in some other thread). Such
a move will always respect thread order but can still be an incorrect
reordering due to the following reasons:

\begin{itemize}
\item A \R{x} does not see the same value because it's relative position w.r.t
a \W{x} has been changed, we call this a RW violation. If such a violation
happens then the two conflicting events $e,e'$ will be such that $c \tr e \tr e'$ and $e' \in [a,b]$ in $\tau$
\item Lock semantics has been violated: this can happen in two ways
\begin{itemize}
\item $[a,b]$ contains a critical section whose lock is already held at $c$, we call this LS1 violation.
\item $[a,b]$ contains an acquire $d = \acq{l}$ but not $\mtc(d')$ such that there
is critical section $(h,h')$ in another thread over the same lock $\l$ such that 
$c \tr h \tho h' \tr d$. We call this a LS2 violation.
\end{itemize}
\end{itemize}

\begin{lem}\label{lem:between} 
For all events $e$ such that $e_1 \tre e \tre e_2$ we have
\begin{enumerate}
\item $e_1 \hb e \hb e_2$ 
\item $e_1 \nlw e$ and $e \nlw e_2$
\end{enumerate}
\end{lem}

\begin{lem}\label{lem2} Let $a_1$ be an acquire event such that $e_1 \in \cs(a_1)$.
For all events $e$ such that $a_1 \tre e \tre e_1$ we have $a_1 \hb e$
and $e \hb e_2$.
\end{lem}

\begin{lem}\label{lem:conflict}
Any conflicting pair of events appearing before $(e_1,e_2)$ have to be
WCP ordered.
\end{lem}

\begin{lem}\label{lem:across}
For any acquire event $a_1$ such that $e_1 \in \cs(a_1)$, and any critical
section $(a,r)$ such that $a \tr a_1 \tr r \tr e_1$ it is the case that
$a_1 \lwcp r$.
\end{lem}

\begin{lem}\label{lem:dangle}
Any acquire event $a$ such that $\mtc(a) \notin \alpha$ then $a$ is in $t_2$.
\end{lem}

Let $f$ be the first event after $e_1$ which is not in $t_1$.  Next note that
by Lemma~\ref{lem:between} we have $e_1 \hb f$ and $f$ being the first event
outside $t$ has to be an acquire, say over lock $\l$. More over there has to
be a $\rel{l}$ event, say $g'$, such that $e_1 \tho g' \tr f$. Let $g'$ be the
last such event. Let $g = \mtc(g')$.

Among all extremal traces, pick $\tau$ such that it minimizes the distance
from $f$ to $e_1$,  If $\mtc(f)$ exists in $\tau$ it shall be referred to as
$f'$. We shall refer to this $\tau$ as a \emph{minimal trace}.  All our
arguments that follow will pertain to this minimal trace $\tau$.

\begin{lem}\label{lem:rw}
If there are events $e,e'$ such that $e \lwcp e'$
and $g \tre e$ and one of the following holds:
(a) $f'\notin \tau$
(b) $e' \hb f'$,
then $e_1 \lwcp e_2$
\end{lem}
\begin{proof}
Consider two sub-cases: $e_1 \tho g$ and $g \tho e_1$.
\begin{enumerate}
\item $e_1 \tho g$: $e_1 \tho g \tre e$ implies $e_1 \hb e$ using Lemma~\ref{lem:between}.
  Then $e_1 \hb e \lwcp e' \hb e_2$ gives $e_1 \lwcp e_2$ using Rule (c) of WCP.
\item $g \tho e_1$: consider further sub-cases
\begin{enumerate}
\item $f' \in \tau$: $g \hb e $ (using Lemma~\ref{lem2} if $e \tr e_1$ or Lemma~\ref{lem:between} if $e_1 \tr e$) and $e \lwcp e' \hb f'$. This
implies $g \lwcp f'$. Applying Rule (b) we get $g' \lwcp f'$. Combining
this with $e_1 \tho g'$ and $f' \hb e_2$ using Rule (c) we get $e_1 \lwcp e_2$.
\item $f' \notin \tau$: $f \in t_2$ (Lemma~\ref{lem:dangle}) and Rule (a) of
WCP gives $e_1 \lwcp e_2$. \qedhere
\end{enumerate}
\end{enumerate}
\end{proof}

{}

\begin{lem}\label{lem:suffix} If segment $[a,b]$ of a thread is moved
right before event $c$ where $c \tr a$ and $[a,b]$ is the suffix of a critical
section then the move cannot violate LS2. \end{lem}
\begin{proof}
Whenever we move a segment of a thread that is the suffix of a crtical section
it is always the case that every acquire in the segment also has its matching release
in the segment. This is due to well-nestedness. In a LS2 violation one requires only
an acquire to be moved up without its release.
\end{proof}

The rest of the proof structure is as follows: we
attempt to show that $\tau$ exhibits a predictable deadlock. In order to
construct these deadlocks we are going to define structures called
\emph{deadlock chains}. We show that presence of a deadlock chain results in a
predictable deadlock through structures called \emph{deadlock patterns}.

\begin{defn}\label{def:conflict1}
A critical section $(c,c')$ appearing after $g$ (i.e., $g \tr c$) is said
to be \textbf{conflicting} due to critical section $(d,d')$, if $(d,d')$ is
over the same lock as $c$ such that $d \tr g \tr d' \hb c$.
\end{defn}

\begin{defn} A \textbf{deadlock chain of type-1} of length $k$ is a sequence
of threads $v_1,\dots,v_{k+1}$ and a sequence of critical sections
$(c_1,c_1'),(d_1,d_1')$ \dots $(c_k,c_k'),(d_k,d_k')$ such that

\begin{enumerate}[label=(\alph*)]
\item $v_1,\dots,v_{k+1}$ are distinct and $t_1 \notin \set{v_1,\dots,v_{k}}$
\item $v_1$ is the thread containing $f$
\item $\forall\, i: (c_i,c_i')$ is the earliest conflicting critical section
contained in $(d_{i-1},d_{i-1}')$, or $\cs(f)$ if $i=1$
\item $\forall\, i: (c_i,c_i')$ is conflicting due to $(d_i,d_i')$
\item $\forall\, i: (d_i,d_i')$ is contained in thread $v_{i+1}$
\end{enumerate}
\end{defn}

\begin{lem}
If $e_1 \in (g,g')$ then $f' \in \tau$
\end{lem}
\begin{proof} Follows from Lemma~\ref{lem:dangle} and Rule (a) of WCP. \end{proof}

\begin{defn}\label{def:conflict2} When $e_1 \in (g,g')$, an acquire event $d$ is
said to be \textbf{conflicting} due to a critical section $(c,c')$, if $g
\tr d \tr e_1$,  $(c,c')$ is over the the same lock as $d$, and $g \tr c \tho c'
\hb d$ \end{defn}

\begin{defn} A \textbf{deadlock chain of type-2} of length $k$ is a sequence
of threads $v_1,\dots,v_{k+1}$ and a sequence of acquire events and critical
sections $d_1,(c_1,c_1'),\dots,d_k,(c_k,c_k')$ such that $e_1 \in (g,g')$ and:
\begin{enumerate}[label=(\alph*)] \item $v_1,\dots,v_{k+1}$ are distinct and
$t_1 \notin \set{v_1,\dots,v_{k}}$

\item $v_1$ is the thread containing $f$

\item $\forall\, i:$ $d_i$ is the earliest conflicting acquire whose
critical section contains $(c_{i-1},c_{i-1}')$, or $(f,f')$ if $i=1$

\item $d_i$ is conflicting due to $(c_i,c_i')$

\item $(c_i,c_i')$ is contained in thread $v_{i+1}$

\end{enumerate}
\end{defn}

Now, we go back to $\tau$ and show presence of deadlock chains. 

\begin{lem}\label{lem:main}
The minimal trace $\tau$ contains a deadlock chain of type-1 or type-2
\end{lem}

\begin{proof}
Assume two subcases:
\begin{enumerate}

\item $f'$ occurs in $\tau$. Let $u$ be the first event after $g$ in the
thread containing $f$. We attempt to move $[u,f']$ to right before $g$. If
this move is successful it violates either: \begin{itemize} \item extremality
by decreasing distance between $e_1,e_2$ (when $g \tho e_1$) \item minimality
by moving the resulting $f$ closer to $e_1$ (when $e_1 \tho g$) \end{itemize}
Hence this move cannot be successful, which cannot be due to RW violation from
Lemma~\ref{lem:rw}. So lock semantics has to be violated, which can happen in
two ways:  

\begin{enumerate} \item If LS1 is violated, there is a conflicting
critical section $(c,c')$ contained in the segment $[u,f']$ due to critical
section $(d,d')$. 
\begin{itemize}
\item Consider $g \tho e_1$: if $c \tr e_1$ then $d' \tr e_1$
and we know from Definition~\ref{def:conflict1} that $g \tr d'$. Applying
Lemma~\ref{lem:across} on $g$ and $d'$ gives us $g \lwcp d'$. But we know $d'
\hb c \tho f'$ which gives us $g \lwcp f'$. Applying Rule (b) of WCP we get $g'
\lwcp f'$. Combining this with $e_2 \tho g'$ gives us $e_2 \lwcp f'$  and we
also know $f' \hb e_2$ (Lemma~\ref{lem:between}) which gives us $e_1 \lwcp e_2$
a contradiction. So assuming $c \tr e_1$ led to a contradiction and hence
we get $e_1 \tr c$, which implies $f \tho c$ because $f$
is the first event outside $t_1$ after $e_1$. By well-nestedness we get
$(c,c')$ contained in $(f,f')$.
And so a type-1 deadlock chain of length 1 is obtained. 
\item Consider $e_1 \tho g$: In this case $u = f$ and hence $(c,c')$
is contained in critical section $(f,f')$ and along with $(d,d')$
gives us a type-1 deadlock chain of length 1.
\end{itemize}
\item If LS2 is violated then there is an acquire $d \in [u,f']$, $match(d)
\notin [u,f']$, and critical section $(c,c')$ in another thread over the same
lock as $d$ such that $g \tr c \toe c' \tr d$. This means $d$ is conflicting
due to $(c,c')$. We also have $e_1 \in (g,g')$, otherwise we get that $u = f$
which implies that $match(d) \in (f,f')$ due to well-nestedness, a
contradiction. And therefore LS2 gives us a type-2 deadlock chain of length 1.
\end{enumerate}

\item $f'$ does not appear in $\tau$. Then $f'$ is performed by $t_2$, with
$e_2 \in \cs(f')$ by Lemma~\ref{lem:dangle}. This implies $e_1 \notin (g,g')$
otherwise by Rule (a) of WCP we have $g' \lwcp e_2$ which combined with $e_1
\tho g'$ gives us $e_1 \lwcp e_2$ a contradiction. Next we obtain a trace
$\tau'$ from $\tau$ by dropping all events below and including $g$ from all
threads other than $t_2$. If $\tau'$ is a correct reordering it violates
extremality as distance between $e_1,e_2$ is reduced. $\tau'$ cannot violate
thread order as suffixes of threads are being dropped. If $\tau'$ has RW
violation then there is some \R{x} event, say $e$ which is undropped in $t_2$
after $g$ and some dropped \W{x}, say $e'$, such that $e' \tr e$. Then by
Lemma~\ref{lem:rw} we have $e_1 \lwcp e_2$ which is a contradiction. So
$\tau'$ has to have a lock semantic violation, which can only be because a
release $d'$ was dropped whose matching acquire $d$ is not dropped and there
is a critical section over the same lock $(c,c')$ occuring later than
$(d,d')$, which is undropped and hence in $t_2$. But $f$ is the first
event in $t_2$ after $e_1$ and since $e_1 \tho g \tr d' \tr c$ we get
that $(c,c')$ is contained in $(f,f')$. This gives rise to type-1
deadlock chain. \qedhere

\end{enumerate}
\end{proof}

What remains to be done is to prove that deadlock chains result in predictable
deadlocks. In order to do so we introduce intermediate structures
called deadlock patterns.

\begin{defn}
A \textbf{type-1 deadlock pattern} of rank $r$ is a sequence of threads
$v_1,\dots,v_{k+1}$ and a sequence of critical sections $(c_1,c_1')$, $(d_1,d_1')$
\dots $(c_k,c_k'), (d_k,d_k')$ such that:
\begin{enumerate}[label=(\alph*)]
\item $v_1,\dots,v_{k+1}$ are distinct and $v_{k+1} = t_1$
\item $v_1$ is the thread containing $f$ {and $(c_1,c_1') \in \cs(f)$}
\item $\forall\, i > 1: (c_i,c_i')$ is contained in $(d_{i-1},d_{i-1}')$ in thread $v_i$. 
\item $\forall\, i: d_i$ and $c_i$ are over the same lock and $d_i' \hb c_i$
\item $\forall\, i \in (r,k]\; \forall\, j: P(i,j)$ (where $P(i,j)$ is defined below)
\item $\forall\, i:$ Any critical section $(c,c')$ such that $g \tr c \tho c_i$
is non-conflicting.

\item $d_{k} \tho g$
\end{enumerate}

where $P(i,j)$ is the property that for any acquire event $h_1$ such that  $c_i \in
\cs(h_1)$ and for any acquire event $h_2$ over the same lock as $h_1$ such
that $h_2 \tho c_j$ then $h_2 \hb h_1$.

\end{defn}

\begin{lem}\label{lem:type1distance}
In a deadlock chain of type-1 if $v_{k+1} \neq t_1$ then $f \tr d_k'$
\end{lem}
\begin{proof}
Follows from Lemma~\ref{lem:across} and the fact that $f$ is first event outside $t_1$.
\end{proof}

\begin{lem}\label{lem:type1}
If there exists a deadlock chain of type-1 then there exists a type-1 deadlock
pattern
\end{lem}
\begin{proof} Define distance between $g$ and $d_k'$ as: 0 if $v_{k+1} = t_1$,
or number of events between $g$ and $d_k'$ (inclusive) otherwise. We perform induction
on this distance.

Base case: If the distance is 0 then $v_{k+1} = t_1$ by
Lemma~\ref{lem:type1distance}. This gives us the desired type-1 deadlock
pattern which has rank $k$. Property (e) of the pattern comes for
free since rank $= k$ and the remaining properties follow from the
properties of the chain.

Inductive case: if the distance $> 0$ then $v_{k+1} \neq t_1$ (definition of
the distance). Let $u$ be the first event in $v_{k+1}$ after $g$. Consider a
move involving $[u,d_k']$ to just before $g$. The move violates extremality or
minimality (either due to decreasing $e_1,e_2$ distance or moving the
resulting $f$ closer to $e_1$). If the move commits a RW violation then there
have to be two conflicting events $e,e'$ such that $g \tr e $ and $e' \hb c_1$
(this can be inductively proved). So if $f' \in \tau$ then $e' \hb f$. Now
applying Lemma~\ref{lem:conflict} and \ref{lem:rw} we get $e_1 \lwcp e_2$ a
contradiction. So the move has to violate lock semantics, but it cannot be a
LS2 violation since $[u,d_k']$ is a suffix of the critical section
$(d_k,d_k')$ and Lemma~\ref{lem:suffix}. Therefore the move has to violate
LS1: which means there is a critical section $(c_{k+1},c_{k+1}')$ in
$[u,d_k']$ which is conflicting. Pick the earliest among all such critical
sections. Let the conflict be due to $(d_{k+1},d_{k+1}')$ in thread $v_{k+2}$.
If $v_{k+2} \notin \set{v_1,\dots,v_{k+1}}$ then this process extends  the
deadlock chain while reducing the distance, and by the induction hypothesis we
have the required deadlock pattern. So consider when $v_{k+2} = v_i$ where $i
< k+1$. We have $d_{k+1}' \tho d_{i-1}'$ which implies $d_{i-1} \tho d_{k+1}$.
Let $a$ be the first event in $v_i$ after $g$. Consider a move of
$[a,d_{k+1}']$ to right before $g$. Once again this move violates
minimality/extremality. As before the move cannot commit a RW violation, and
lock semantics can be violated only due to LS1. Hence there is a critical
section $(c_{k+2},c_{k+2}')$ inside $[a,d_{k+1}']$ that is conflicting, but
this violates the choice of $(c_i,c_i')$ as the earliest conflicting critical
section contained in $(d_{i-1},d_{i-1}')$, property (c) of the deadlock chain.
\end{proof}

Similarly we define type-2 deadlock patterns which we will derive from
deadlock chains of type-2.

\begin{defn}
A \textbf{type-2 deadlock pattern} of rank $r$ is a sequence of threads
$v_,\dots,v_{k+1}$ and a sequence of acquires and critical sections
$d_1,(c_1,c_1')$ \dots $d_k,(c_k,c_k')$ such that $e_1 \in (g,g')$ and:
\begin{enumerate}[label=(\alph*)]
\item $v_1,\dots,v_k$ are distinct and $v_{k+1} = t_1$
\item $v_1$ is the thread containing $f$ and $f \in \cs(d_1)$
\item $\forall\, i \le k: (c_i,c_i')$ is contained in $(d_{i+1},d_{i+1}')$ in thread $v_{i+1}$.
\item $\forall\, i: d_i$ and $c_i$ use the same lock and $g \tr c_i \hb d_i \tr e_1$
\item $\forall\, i > r: \forall\, j: Q(i,j)$ (where $Q(i,j)$ is defined below)
% \item $\forall\, i:$ Any acquire $a$ such that $g \tr a \tho d_i$ is non-conflicting
\item $\forall\, i \le r: d_i' = \mtc(d_i)$ if present in $\tau$ then $e_1 \tr d_i'$
\item $(c_{k},c_k') \in (g,g')$
\end{enumerate}

where $Q(i,j)$ is the property that for any acquire event $h_1$ such that
$c_{i-1} \in \cs(h_1)$ and for any acquire event $h_2$ such that $h_2 \tho c_j$
then $h_2 \hb h_1$.
\end{defn}

\begin{lem}\label{lem:type2}
If there exists a deadlock chain of type-2 then there exists a type-2 deadlock
pattern
\end{lem}
\begin{proof}
Define size of a deadlock chain of type-2 as the distance between $g$ and $c_k$
if $v_{k+1} \neq t_1$ and 0 otherwise. We perform induction on the size of the chain.

Base case: When size = 0 we get $v_{k+1} = t_1$ and we get the chain to be a
pattern of order $k$. Property (f) of the pattern comes for free since rank $=
k$ and the remaining properties follow from the properties of the chain.

Inductive case: If size $> 0$ then $v_{k+1} \neq t_1$. Let $u$ be the first
event in $v_{k+1}$ after $g$. Consider a move of $[u,c_k']$ to just before
$g$. If move is successful it violates minimality, hence the move is invalid
As before the move cannot commit a RW violation. Therefore the move is either
a LS1 or LS2 violation. First, we show why a LS1 violation is not possible. If
it does then there exists a critical section $(c_{k+1},c_{k+1}')$ across $g$
and a moved acquired, say $a$, over the same lock in thread $v_{k+1}$. The
critical section $(c_{k+1},c_{k+1}')$ cannot be in $t_1$ because by well
nestedness we then get $g' \tho c_{k+1}'$ and therefore $e_1 \tho c_{k+1}'$,
which makes it impossible for $a$ to exist as it should be between $g$ and
$e_1$. If $(c_{k+1},c_{k+1}')$ is performed by $v_{k+1}\neq t_1$ then we get
that $c_{k+1} \tr g \tr c_{k+1}' \tr e_1$. Now we apply Lemma~\ref{lem:across}
on $g$ and $(c_{k+1},c_{k+1}')$ to get $g \lwcp c_{k+1}'$, that combined with
$c_{k+1}' \hb a \tho c_k' \hb f'$ gives us $g \lwcp f'$ and applying
Lemma~\ref{lem:rw} gives us $e_1 \lwcp e_2$, a contradiction. So we are left
with LS2 violation as the only possibility in which case there is a moved
acquire $d_{k+1}$ which is not released before $c_k'$ and a critical section
$(c_{k+1},c_{k+1}')$ over the same lock as $d_{k+1}$ such that $g \tr c_{k+1}
\tho c_{k+1}' \hb d_{k+1}$. Among all such $d_{k+1}$ we choose the earliest.
Now if $v_{k+2}$ so obtained is $t_1$ we are done because size of the extended
chain becomes a zero yielding the base case. If $v_{k+2} \neq t_1,v_1,v_2
,\dots,v_k$ then we would have extended the chain while have reduced the size
of the chain since $c_{k+1}$ is earlier than $c_k$ in $\tau$. Consider what
happens when $v_{k+2} = v_i$ for some $i \le k$. Let $a$ be the first event in
$v_i$ after $g$. We attempt to move $[a,c_{k+1}']$ to right before $g$. The
move cannot be successful because it would violate minimality of $g$. Now
using reasoning similar to above we get that RW and LS1 violations are not
possible, leaving us with LS2 violation. So there exists an acquire $d_{k+2}$
such that $g \tr d_{k+2}$ and $(c_{k+1},c_{k+1}') \in \cs(d_{k+2})$. Choose
$d_{k+2}$ to be the earliest such acquire. If $d_{k+2}' = \mtc(d_{k+2})$ is
not in $\tau$ then $\cs(d_{k+2})$ contains $(c_{i-1},c_{i-1}')$ violating the
optimal choicef $d_i$ in thread $v_i$. If $d_{k+2}' \in \tau$ and later than
$c_{i-1}'$ and the result is similar. Consider the case when $d_{k+2}'$ is
earlier than $c_{i-1}'$ in which case $d_{k+2}' \tho c_{i-1}$ by well
nestedness. THen consider moving segment $[a,d_{k+2}']$ to just before $g$.
Again the move has to be unsuccessful as it violates minimality, and using
reasoning exactly the same as above we can eliminate violations due to RW/LS1.
Now in this case LS2 cannot happen either because if it does we get an acquire say $a'$ held at $d_{k+2}'$ (i.e., $d_{k+2}' \in \cs(a')$ and $a' \tho d_{k+2}$)
such that there is a critical section $(b,b')$ (in some thread $\neq v_i$)
over the same lock as $d_{k+2}$ and $g \tr b$. But this $a'$ violates
the chocie of $d_{k+2}'$ being chosen as the earliest. So we conclude
that $v_{k+2} \notin {v_1,\dots,v_{k+2}}$ and the inductive step goes through.
\end{proof}

Next we look at how deadlock patterns can be used to show predictable
deadlocks.

\begin{lem}\label{lem:pred1}
If there exists a type-1 deadlock pattern then of any rank there exists a predictable deadlock
\end{lem}

\begin{proof} We perform induction on the rank $r$.

Base case: rank $r = 0$. Consider the execution $\tau'$ obtained by dropping
the following events from $\tau$: (i) events $e$ s.t $c_i \toe e$ (ii) events
$e$ s.t $g\toe e$ (iii) events $e$ in threads $\notin \set{v_1,\dots,v_{k+1}}$
s.t $g \tr e$. The claim is $\tau'$ is a correct reordering of $\tau$. $\tau'$
reveals a deadlock.  Proof of claim: $\tau'$ does not violate thread order as
we are only dropping suffixes of threads. The read events that possibly don't
see the same writes is because the write events have been dropped, implying
both the read event $e'$ and the write $e$ are  later than $g$, we then
observe that if $f' \in \tau$ then $e' \hb f'$ and  apply Lemma~\ref{lem:rw}
to get a $e_1 \lwcp e_2$, a  contradiction. But $\tau'$ may not valid trace
due to lock semantic violation. So there exists a critical section $(a_1,r_1)$
and acquire $a_2$ of the same lock such that $r_1 \hb a_2$, $r_1$ is dropped
and $a_1,a_2$ are undropped. Since $r_1$ is dropped we get $g \tr r_1$ and so
$g \tr a_2$ which implies $a_2$ is performed by some $v_i$.  Suppose $a_1$ is
not peformed by some $v_j$ then since $a_1$ is undropped we know $a_1 \tr g$,
but this implies $(a_2,\mtc(a_2))$ conflicts with $(a_1,r_1)$
contradicting the choice of $(c_i,c_i')$ as the earliest conflicting critical
section (Note that $\mtc(a_2)$ exists in $\tau$ because $a_2
\in (d_{i-1},d_{i-1}')$ and well-nestedness). So we get $a_2$ is performed by some thread $v_j$.
 Consider any one such pair $a_1,a_2$. Now applying property (e) with
$h_1 = a_1$ and $h_2 = a_2$ we obtain $a_2 \hb a_1$, a contradiction since we
began with $a_1 \hb a_2$ and $a_1 \neq a_2$. We are able to apply property (e)
to $a_1,a_2$ as rank is 0. This implies lock semantics is not violated and
hence $\tau'$ is valid.

Inductive case: Obtain $\tau'$ from $\tau$ as in the base case and follow the
same argument until the application of property (e). Consider among all
possible $a_1,a_2$ those where $a_1$ is in $v_j$ for the largest $j$ and among
those pick pairs where $a_1$ is in thread $v_j$ for the largest $j$, and among
those pick pairs where $a_2$ is in thread $v_i$ for the smallest $i$ and among
those pick the one with the earliest $a_2$ in $v_i$. First we observe $j \le
r$ which can be derived from property (e) and the fact that $a_1 \hb a_2$. The
observation $a_1 \hb a_2$ can also be used to deduce $i < j$. Next we
construct a new deadlock pattern that uses the threads
$v_1,\dots,v_i,v_j,\dots,v_{k+1}$, i.e. we drop threads between $v_i$ and
$v_j$ and use crtical section $(a_2,\mtc(a_2))$ in place $(c_i,c_i')$ and
critical section $(a_1,\mtc(a_1))$ in place of $(d_{j-1},d_{j-1}')$. This
surgery will result in a deadlock pattern of rank $i$ but $i < j \le r$, hence
the rank is reduced and the inductive hypothesis takes care of the rest.
\qedhere \end{proof}

\begin{lem}\label{lem:pred2}
If there exists a type-2 deadlock pattern then of any rank there exists a predictable deadlock
\end{lem}

\begin{proof}
As for type-1 deadlock patterns we perform induction on the rank.

Base case: rank $r = 0$. Consider once again the execution $\tau'$ obtained by
dropping: (i) events $e$ s.t $c_i \toe e$ (ii) events $e$ s.t $f \toe e$.
(iii) events $e$ in threads $\notin
\set{v_1,\dots,v_{k+1}}$ s.t $g \tr e$. The claim is $\tau'$ is a correct
reordering of $\tau$, in which case the ordering reveals a deadlock. Proof of
claim: $\tau'$ can only be invalid trace due to lock semantic violation (same
argument as type-1). Once again we get a critical section $(a_1,r_1)$ and
acquire $a_2$ on the same lock such that $r \hb a_2$, $r_1$ is dropped and
$a_1,a_2$ are undropped. Since $r_1$ is dropped we know $g \tr r_1$ and so $g
\tr a_2$ which implies $a_2$ is performed by some $v_i$. If $a_1 \tr g$ then
using Lemma~\ref{lem:across} and Lemma~\ref{lem:rw} we get a contradiction. So
$g \tr a_1$ and hence $a_1$ is performed by some $v_j$. In the base case rank
being 0 we can apply property (e) with $h_1 = a_1$ and $h_2 = a_2$ to get a
contradiction $a_2 \hb a_1$. This implies lock semantics is also not violated
and hence $\tau'$ is valid.

Inductive case: Follow the same steps as in the base case until obtaining
$a_1,a_2$. Then observe that $a_1$ has to be performed after $d_j$ and before
$c_{j-1}$, otherwise moving the segment of that thread after $g$ until $r_1$
to above $g$ leads to conclusion that $d_j$ was not chosen as the earliest a
contradiction. $a_1$ should be after $d_i$ otherwise $d_i$ being earliest
critical section is violated due to presence of $(a_1,r_1)$. Now we short
circuit the deadlock pattern by considering threads
$v_1,\dots,v_i,v_j,\dots,v_{k+1}$, replace $d_i$ with $a_2$ and $d_j$ with
$a_1$. This will result in a deadlock pattern of rank $i (< j \le r)$ and then
the inductive hypothesis can be applied.
\end{proof}

Using Lemma~\ref{lem:main} we get that $\tau$ contains a deadlock chain. Using
Lemmas~\ref{lem:type1} and \ref{lem:type2} we show presence of deadlock
patterns, and finally using Lemmas~\ref{lem:pred1} and \ref{lem:pred2} we get
the required predictable deadlock, thus proving Theorem~\ref{thm:sound}.

\section{Errors in CP soundness proof}\label{sec:cperrors}

For the reader familiar with CP soundness proof \cite{cp2012} we point out
some of the errors in it.

\begin{itemize}
\item In case 1 (a) of the proof of the main theorem when $f\dots f'$
is being moved to $g$ and is inspected for a lock semantic violation
they argue \emph{``If $m$ is held at a point $g$ by a thread $t_3$ other than $t_1$ then it has to be release before $f$''}. This is incorrect
because the lock $m$ can be released by $t_3$ after $f$ and before the acquire of $m$
present in the critical section $f\dots f'$.
\item In case 1 (b) of the proof of the main theorem when $u\dots f'$
is being moved to $a_2$ (where $u$ is the first event after $a_2$ in thread
containing $f$) and lock semantic violation is explored they say
\emph{``If such an $m$ is held by a thread $t_3$, other than $t_1$ at point $a_2$,
then it has to be released before $e_1$''}. This is not true because
the $m$ can also be released after $e_1$ in which case the release
will have to be after $f$ which is not a impossibility. Further more
this is not the only way in which lock semantics can be violated.
As we have shown in our proof there are two ways of lock semantic violation
when moving up a segment of a thread: LS1 and LS2. The above argument
only explores LS1 violation (that too partially). LS2 violation
is completely ignored in their proof.
\item In case 2 of the proof of the main theorem: after dropping
the events to obtain the trace $tr'$, when they argue lock semantics violation
they say \emph{``The acquisition of $m$ has to be in thread $t_1$ (otherwise $f$ would
not be the first event ...)''}. This is wrong because the acquisition
of $m$ can be in another thread, the only thing is that the it's corresponding
release that is dropped has to be after $f$. Nothing prevents it from
being after $f$.
\end{itemize}

\section{Correctness of algorithm}\label{app:algo}
 Given a trace $\sig$, assume
that we have updated the vector clocks as described in
Algorithm~\ref{algo:update} and times $C_e,H_e,P_e,N_e$
have been assigned to event $e$ as described in Section~\ref{sec:algo}.

\begin{lem}\label{lem:lastrelease}
At any point:
$\H_{\ell}=H_r,\P_\ell=P_r$ where $r$ is the last $\rel{\ell}$ event.
$\L_{\ell,x}^r = \mx_{r} H_r$ is where $r$ ranges over all $\rel{\ell}$
event such that $\cs(r)$ contains a \R{x} event. 
$\L_{\ell,x}^w = \mx_{r} H_r$ is where $r$ ranges over all $\rel{\ell}$
event such that $\cs(r)$ contains a \W{x} event. 
\end{lem}
\begin{proof}
Can be observed by the fact that these clocks are only updated during
a release. And the respective invariants can be easily checked from Lines 9,7,8.
\end{proof}

\begin{lem}\label{lem:threadmonotone} 
If $a \toe b$ then $C_a \cle C_b$, $P_a \cle P_b$, $H_a \cle H_b$
\end{lem}
\begin{proof}
Follows from the fact that each time clocks $\P_t,\H_t$ are 
updated (Lines 1,2,6,11,12) they are assigned clocks which take maximum of their previous value with some other clock. Also the counter $N_t$ is only
incremented.
\end{proof}

% \begin{lem}\label{lem:hbimplieshbclock}
% For any two events $a,b$ if $a \hb b$ then 
% \end{lem}
% \begin{proof}
% \end{proof}

\begin{lem}\label{lem:hbimpliesclock} 
For any two events $a,b$ if $a \hb b$ then $H_a \cle H_b$ and $P_a \cle P_b$
\end{lem}
\begin{proof}

If $a \toe b$ then by Lemma~\ref{lem:threadmonotone} gives us the required
result. Now suppose $a,b$ are in different threads. We perform induction on
the position of $b$ in the trace. In the base case when $b$ is in the first
position in the trace we have $a = b$ in which case the result trivially
follows. Suppose $b$ is not in the first position in the trace. Consider the
event $c$ preceding $b$ in its thread (it may not exist). If $a \hb c$ we use
the induction hypothesis along with the Lemma~\ref{lem:threadmonotone} to get
our result. Otherwise $b$ has to be an $\acq{l}$ event, and there has to be a
previous $\rel{l}$ event. Call the last such release before $b$ as $r$. Now $a
\hb r$, and by the induction hypothesis we get
$H_a \cle H_r$ and
 $P_a \cle P_r$. Note that when $r$ updated the state, $\H_l$ and
$\P_l$ were assigned $H_r$ and $P_r$ respectively in Line 9. And when
$b$ is used to update the state $H_r$ was joined into $\H_t$ in Line 1
and $P_r$ was joined into $P_b$ in Line 2.
Hence $H_r \cle H_b$ and $P_r \cle P_b$, combining this with inequalities
obtained from applying the induction hypothesis on $a$ and $r$ gives us our
required result.
\end{proof}

\begin{lem}\label{lem:ph}
For any event $e$, $P_e \cle C_e \cle H_e$
\end{lem}

\begin{proof}
For $P_e \cle C_e$ note that $C_e$ is just a short hand for $P_e[t(e) := N_e]$
and the fact that no vector clock until $e$ can have it's $t(e)$ component's
value bigger than $N_e$.

For $C_e \cle H_e$ it is sufficient to prove $P_e \cle H_e$ 
because $P_e$ and $C_e$ match on every co-ordinate except possibly
$t(e)$, but for $t(e)$ we know $P_e(t(e)) = H_e(t(e)) = N_e$.
 For $P_e \cle H_e$ we do the following: Any thread begins with
$\P_t$ and $\H_t$ being equal to the 0 vector-clock. As each event is
processed we prove that the $\P_t \cle \H_t$ is maintained inductively. We do
a case analysis on the type of events. If $e$ is a an acquire, the from Lines
1 and 2 and the induction hypothesis we get that $\P_t \cle \H_t$ is
maintained. 

When $e$ is a $\rel{l}$ event: $\P_t$ is updated at Line 6, which has the time
$H_r$ of a release event $\rel{l}$, say $r$ that is performed by thread
$t'$. But note that when  $a = \mtc(e)$ was performed $\H_t$ received the
time from clock $\H_{l}$ (Line 1) which corresponds to the last release on lock $l$, say
$r'$. Clearly $r'$ is later than $r$ in the trace and since they operate over
the same lock we have $r \hb r'$ and applying Lemma~\ref{lem:hbimpliesclock}
we get $H_{r} \cle H_{r'}$. In Line 1 when $a$ is being
processed $\H_t$ is updated with $\H_l$ which is nothing but $H_{r'}$. Which means by the time Line 6 updates $\P_t$, the value of
$\H_t$ is already updated with a larger time.

Next consider the case when $e$ is either a \R{x}/\W{x} in which case $\P_t$
gets updated in Lines 11/12, note that in both cases it gets the time
$\L_{\l,x}^{r/w}$ of a certain $\rel{l}$, say $r$, only if the event $e$ is being
performed inside a critical section of $\l$, let the acquire event of that
critical section be $a$. Now we know that when $a$ is performed it receives
the the time $H_{r'}$ (through $\H_l$) of last $\rel{l}$ event
$r'$ before $a$. Now once again we have $r \hb r'$ and using the same argument
as in the previous paragraph we have our inequality maintained.

\end{proof}

\begin{lem}\label{lem:releasecounter}
For any two events $e,r$ in a thread $u$ where $r$ is a release,
if $N_a \le N_r$ then $a \toe r$.
\end{lem}
\begin{proof}
Each time a release is processed the counter is incremented prior to the next 
event, this assigns a strictly greater counter for all subsequent events 
in that thread.
\end{proof}

% \begin{lem}\label{lem:releasetime}
% For any event $e$ in thread $u$ and a $\rel{l}$ event r, if $N_a \le C_r(u)$ then either $N_a \le H_r(u)$ or $a \toe r$
% \end{lem}
% \begin{proof}
% First note that $C_r$ is just a shorthand for $P_r[t(r) := N_r]$. Using 
% Lemma~\ref{lem:ph} we have $P_r \cle H_r$ which implies that if
% $t(r) \neq u$ then $N_a \le H_r(u)$. Suppose $t(r) = u$ then we have
% $N_a \le N_r$ for which we apply Lemma~\ref{lem:releasecounter}.
% \end{proof}

\begin{lem}\label{lem:hbp}
For any two events $a,b$ in different threads $u,v$, if $N_a \le H_b(u)$
then $a \hb b$
\end{lem}

\begin{proof}

 We prove this by induction on the position of the event $b$ in the trace. In
the base case we have $b$ in the first position in trace and hence in $v$, so
we get $H_b(u) = 0$. We also have $N_a \ge 1$. Therefore the implication is
vacuously true. Consider the inductive step: let $b$ be not the first position
in the trace . We only need to consider the case where $b$ is not in the first
position in the threads as well (otherwise handled like base case). Let $c$ be
the event just before $b$ in $v$. If $N_a \le H_c(u)$ then we apply induction
hypothesis on $a,c$ to obtain $a \hb c$ and then use $c \tho b$ to get $a \hb
b$. Now suppose event $c$ does not satisfy the requirements of the induction
hypothesis $N_a \nleq H_c(u)$, then we get that $b$ is an acquire event,
because those are the only events at which $\H_t$ changes. The way $\H_t$ is
updated in Line 1 tells us that either $N_a \le H_r(u)$ where $r$ is the last
$\rel{l}$ event. If $a$ and $r$ are in different threads then  we can apply
the induction hypothesis on $r$ and use the fact that $r \hb b$ (definition of
HB) to obtain $a \hb b$. If $a$ and $r$ are in the same thread then we have
$N_a \le N_r$ and we apply Lemma~\ref{lem:releasecounter} to get
$a \tho r$ which implies $a \hb r$.
\end{proof}

\begin{lem}\label{lem:release}
For any event $a$ in thread $u$ and a release event $r$
if $N_a \le H_r(u)$ then $a \hb r$
\end{lem}
\begin{proof}
If $r$ is in thread $u$ we have $H_r(u) = N_r$ in which
case we can apply Lemma~\ref{lem:releasecounter} to get $a \toe r$ and hence
$a \hb r$. If not then we apply Lemma~\ref{lem:hbp} to get $a \hb r$.
\end{proof}

\begin{lem}\label{lem:clockimplieswcp}
For any two events $a,b$ in threads $u,v$ if $N_a \le P_b(u)$
then $a \lws b$. %and \hl{$P_a \cle P_b$} (is this needed?).
\end{lem}
\begin{proof}

We prove this by induction on position of $b$ in the trace. In
the base case since $b$ is in the first position in $v$ we get $P_b(u) = 0$. We
also have $N_a \ge 1$. This makes the required implication vacuously true.

Consider the inductive step. Let $c$ be the event just before $b$ in $u$. If
$N_a \le P_c(u)$ then applying induction hypothesis on $a,c$ we get the required
result. Suppose $P_c(u) < N_a \le P_b(u)$ then we get that $b$ is an event where 
$\P_t$ is updated. We do case analysis on different kinds of events for $b$

\begin{itemize}

\item Suppose $b$ is $\acq{l}$ then Line 2 tells us that $N_a \le P_r(u)$ where
$r$ is the last $\rel{l}$ before $b$. Then we apply induction on $a,r$ to get $a
\lwcp r$. Using $r \hb b$
(Definition of HB) and applying Rule (c) of WCP
we get $a\lwcp b$.

\item Suppose $b$ is a \R{x} then Line 11 tells us that for some $\ell$,  $N_a
\le H_r(u)$ where $r$ is some $\rel{\ell}$ event such that $b \in \ell$
and $\cs(r)$ contains a \W{x}. Applying
Lemma~\ref{lem:release} on  $N_a \le H_r(u)$ we get $a \hb r$. Combining this with $r \lwcp b$ (Rule
(b)) we get $a \lwcp b$ using Rule (c). 

\item If $b$ is a \W{x} the argument is similar to the case of \R{x} above.

\item Suppose $b$ is a $\rel{l}$ event. Note that $\P_t$ is updated in Line 6
iteratively. Given $b$ we prove that for
events $a$ for which $N_a \le P_b(t(a))$ after $i$ iterations ($i$ times the 
while loop has been executed) $a\lwcp b$ holds by inducting on $i$. 
The base case when
$i = 0$ implies the previous event $c$ is such that $N_a \le P_c(u)$ and by
the outer induction (on position of $b$) gives us that $a \lwcp c$ and from
thread order we have $c \hb b$, composing the two we get $a\lwcp b$.

 Suppose $N_a\le P_b(u)$ after $i$ but not $i-1$ iterations. This implies
there exists an acquire event $e$ and it's matching release $r = \mtc(e)$ such
that $N_a \le H_r(u)$ (Line 6 and Lemma~\ref{lem:lastrelease}). Once again if
$N_a \le H_r(u)$ then we have $a \hb r$  by Lemma~\ref{lem:release}.   We also
have that value of $C_b$ at the end of $i-1$ iterations is such that $C_e \cle
C_b$ (Line 4). Which implies $N_e \le C_b(t(e))$. But note that the queue
$Acq_\ell(t)$ only contains acquire times  of threads other than $t$ (from
Line 3). This implies $C_b(t(e)) = P_b(t(e))$ and we have $N_e \le P_b(t(e))$
and by  the inner induction hypothesis we have $e \lwcp b$. Using Rule (b) of
WCP we get $r \lwcp b$ and composing the two using Rule (c) we get $a \lwcp
b$. \qedhere \end{itemize} \end{proof}

\begin{cor}\label{cor:clockimplieswcp}
For any two events $a,b$ such that $a \tr b$ if $C_a \cle C_b$ then $a \wcp b$.
\end{cor}
\begin{proof}
If $a$ and $b$ are in the same thread then we have $a \tho b$ and so
 we get $a \wcp b$.
If $a$ and $b$ are in different threads $u$ and $v$, then $N_a = C_a(u) \le C_b(u) = P_b(u)$ or $N_a \le P_b(u)$ and then by using Lemma~\ref{lem:clockimplieswcp}
we get $a \lwcp b$.
\end{proof}

Next we see how to prove the other side of the correspondence.

\begin{lem}\label{lem:wcpimpliesclock}
For any two events $a,b$ if $a \lwcp b$ then $H_a \cle P_b$
\end{lem}
\begin{proof}
We look at how $a \lwcp b$ is derived from the Rules of WCP and
perform induction on this derivation/proof. The base case is when
the derivation is of size one i.e., it is derived from Rule (a)
of WCP, in which case $a$ is a $\rel{l}$ such that $\cs(a)$ contains
an event $e$ which conflicts with $b$ which occurs later inside
a different critical section of $\tt{l}$. When $a$ updates the state,
the clocks $\L_{l,x}^w/\L_{l,x}^r$ (depending upon $e$ being \R{x}/\W{x})
are joined with $H_a$ in Lines 7/8. Later when $b$ updates the state it
reads from $\L_{l,x}^w/\L_{l,x}^r$ (Line 11/12) and assigns it to $P_b$
 and therefore $H_a \cle P_b$

The inductive case involves considering the last step of the derivation which
could be the use of Rule (b) or (c). Let us consider Rule (b) first. Here $a$
and $b$ are two $\rel{l}$ events such that their critical sections contain
events $e_1,e_2$ such that $e_1 \lwcp e_2$. By the induction hypothesis we
have $H_{e_1} \cle P_{e_2}$. Let $d$ be the acquire event
$\mtc(a)$. Since $d \toe e_1$ applying Lemma~\ref{lem:threadmonotone} we get
$C_d \cle C_{e_1}$. Combining these with $P_{e_2} \cle C_{e_2} \cle C_{b}$
(from Lemmas \ref{lem:ph} and \ref{lem:threadmonotone}) we get $C_d \cle C_b$.

Before moving ahead with the above argument we elaborate how the queues are
manipulated. Each $\acq{l}$ event performed by a thread $t$ has its time $\C_t$
inserted into the stack $Acq_l(u)$ for every $u \neq t$ (Line 3) and
similarly for release the time $H_t$ is inserted into the the queue
$Rel_l(u)$ (Line 10). These times are inserted in chronological order
(earlier first) as these events are seen, and they are inserted exactly once.
The only event that prompts deque from the queue $Acq_l(u)$ and
$Rel_l(u)$ are $\rel{l}$ events in thread $u$. Entries
from the two queues are always dequed together to ensure that if a acquire is
removed then its matching release is also removed. Note that if an
acquire/release pair $a'/r'$ has already been removed from the respective
queues during a previous release event $r_1$ in thread $u$ then we can
inductively obtain that $H_{r'} \cle P_{r_1}$ and use $P_{r_1} \cle
P_r$ (Lemma~\ref{lem:threadmonotone}) to get $H_{r'} \cle P_r$. The
base case of this induction involves considering the situation when a pair
$a'/r'$ is removed during the state update of the current release $r$ (and not
a previous release $r_1$ in $u$). Such an $a'$ is removed if and only if
$C_{a'} \cle C_r$ (Line 4), and if it is removed then Line~6 ensures $H_{r'} \cle P_r$.

Continuing our argument in the paragraph before now we get that $H_a
\cle P_b$ because the matching acquire of $a$ which is $d$ is such that $C_d
\cle C_b$. This completes the induction step corresponding to Rule (b) being
the last step of the derivation.

Now the final step is to consider the induction step corresponding to last
step of the derivation being Rule (c). First consider when there exists $c$
such that $a \lwcp c \hb b$ then applying induction hypothesis on $a\lwcp c$
we have $H_a \cle P_c$ and then applying Lemma~\ref{lem:hbimpliesclock}
on $c \hb b$ we get $P_c \cle P_b$ and using transitivity of $\cle$
we get $H_a \cle P_b$. Next consider the other case when there exists
$c$ such that $a \hb c \lwcp b$ then applying Lemma~\ref{lem:hbimpliesclock}
on $a \hb c$ we get $H_a \cle H_c$, and applying the induction
hypothesis on $c \lwcp b$ gives us $H_c \cle P_b$. Once again using
transitivity of $\cle$ we obtain $H_a \cle P_b$.
\end{proof}

\begin{cor}\label{cor:wcpimpliesclock}
For any two events $a,b$ if $a\wcp b$ then $C_a \cle C_b$
\end{cor}
\begin{proof}
If $a \toe b$ then we have $C_a \cle C_b$ from Lemma~\ref{lem:threadmonotone}.
Otherwise we have $a \lwcp b$ by definition of $\wcp = (\lwcp \cup \toe)$,
and applying Lemma~\ref{lem:wcpimpliesclock} we get $H_a \cle P_b$ but
we have $P_b \cle C_b$ and $C_a \cle H_a$ (Lemma~\ref{lem:ph}) which gives us $C_a \cle C_b$ (transitivity of $\cle$).
\end{proof}

Theorem~\ref{thm:algorithm} follows from Corollaries~\ref{cor:clockimplieswcp}
and \ref{cor:wcpimpliesclock}.

\section{Running time analysis of algorithm}\label{app:time}
First note that the join operation on two
vector clocks takes $O(\Nt)$ time (assuming arithmetic can be done in constant
time), as this amounts to taking the pointwise maximum across $\Nt$ different
components of the vector clock. Each call to the procedure $\facq$ takes
$O(\Nt^2)$ time (Lines 1 and 2 takes $O(\Nt)$, Line 3 takes $O(\Nt^2)$). For
the $\frel$ procedure, we analyze the total running time of all the invocation
of $\frel$ procedure. 
Note that the total running time of  the while loop  for all the $\frel$ calls
simply depends on the number of entries (acquire/release times) removed from
the queues, which can be bounded by the number of such entries ever added to
the queues. Note that an entry corresponding to a critical section is added to
$\Nt - 1$ queues (Line 3). Hence the while oopin total takes at most $\Ne\Nt$
(as there are at most $\Ne$ acquire/release events in the trace) steps and each
step takes $O(\Nt)$ time since we are dealing with vector clocks, therefore
the while loop takes $O(\Ne\Nt^2)$ time in total across all $\frel$ calls. A
similar aggregate analysis can be done for Lines 7 and 8 to conclude that they
take $O(\Ne\Nt)$ time in total. Line 9 takes $O(\Ne\Nt)$ in total. The for
loop on Line 10 takes $O(\Ne\Nt^2)$ in total.  The total time spent in
processing acquires is therefore $O(\Ne\Nt^2)$. Each $\read/\wr$  procedure
take time proportional to the number of critical sections it is contained in,
which can be at most $\Nl$ and hence the total time spend in them amounts to
$O(\Ne\Nl)$.  The local clock increment not mentioned in the pseudocode takes
takes $O(\Ne)$ time in total (each individual increment in assumed to be
constant) In total therefore the running time of the vector clock algorithm
comes to $O(\Ne(\Nt^2 + \Nl))$.

\section{Lower Bounds}\label{app:lowerbounds}

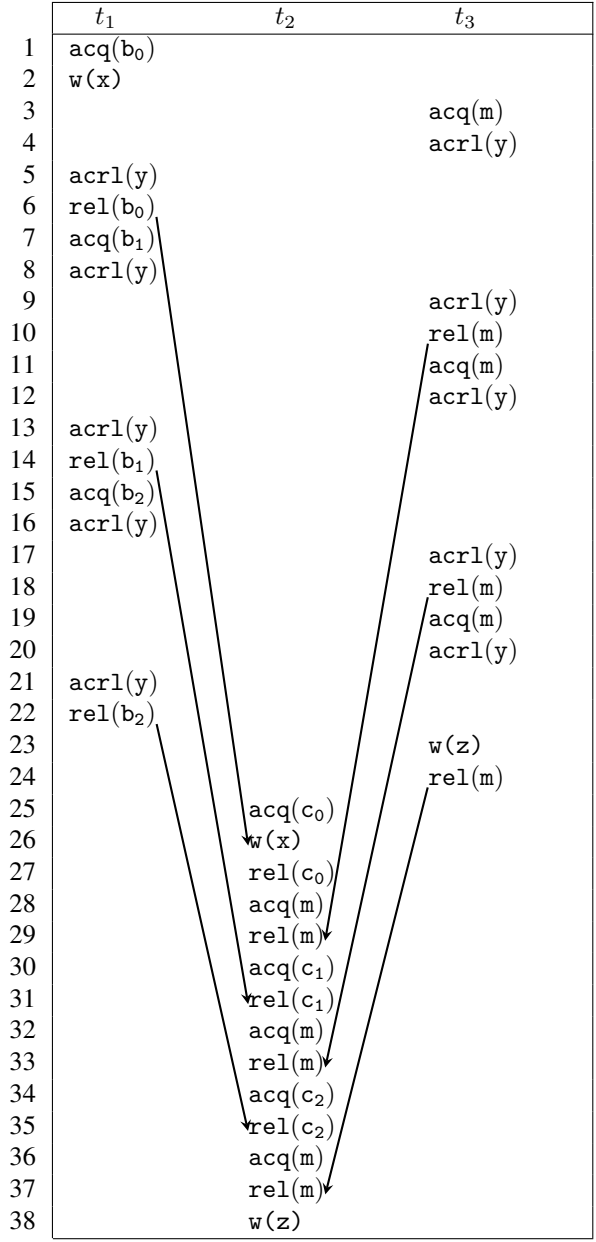
\begin{figure}[h]
\centering
\begin{tabularx}{0.45\textwidth}{r|XXX|}
\cline{2-4}
 & \quad $t_1$ & \quad $t_2$ & \quad $t_3$\\
\cline{2-4}
1 & $\acq{b_0}$ & & \\
2 & \W{x} & & \\
3 & & & $\acq{m}$ \\
4 & & & $\acrl{y}$ \\
5 & $\acrl{y}$ & & \\
6 & $\rel{b_0}$\tikzmark{b0} & & \\
7 & $\acq{b_1}$ & & \\
8 & $\acrl{y}$ & & \\
9 & & & $\acrl{y}$ \\
10 & & & \tikzmark{m0}$\rel{m}$ \\
11 & & & $\acq{m}$ \\
12 & & & $\acrl{y}$\\
13 & $\acrl{y}$ & & \\
14 & $\rel{b_1}$\tikzmark{b1}& & \\
15 & $\acq{b_2}$& & \\
16 & $\acrl{y}$ & & \\
17 & & & $\acrl{y}$ \\
18 & & & \tikzmark{m1}$\rel{m}$ \\
19 & & & $\acq{m}$ \\
20 & & & $\acrl{y}$\\
21 & $\acrl{y}$ & & \\
22 & $\rel{b_2}$\tikzmark{b2}& & \\
23 & & & \W{z}\\
24 & & & \tikzmark{m2}$\rel{m}$ \\
25 & & $\acq{c_0} $& \\
26 & & \tikzmark{wx}\W{x} & \\
27 & & $\rel{c_0}$ & \\
28 & & $\acq{m}$ & \\
29 & & $\rel{m}$\tikzmark{em0} &  \\
30 & & $\acq{c_1}$ & \\
31 & & \tikzmark{c1}$\rel{c_1}$ & \\
32 & & $\acq{m}$ & \\
33 & & $\rel{m}$\tikzmark{em1} & \\
34 & & $\acq{c_2}$ & \\
35 & & \tikzmark{c2}$\rel{c_2}$ & \\
36 & & $\acq{m}$ & \\
37 & & $\rel{m}$\tikzmark{em2} & \\
38 & & \W{z} & \\
\cline{2-4}
\end{tabularx}
\caption{Example trace for showing linear space lower bound.}\label{fig:lowerbound}
\begin{tikzpicture}[overlay, remember picture, >=stealth, shorten >= 1pt, shorten <= 1pt, thick]
    \draw [->] ({pic cs:b0}) to ({pic cs:wx});
    \draw [->] ({pic cs:m0}) to ({pic cs:em0});
    \draw [->] ({pic cs:b1}) to ({pic cs:c1});
    \draw [->] ({pic cs:m1}) to ({pic cs:em1});
    \draw [->] ({pic cs:b2}) to ({pic cs:c2});
    \draw [->] ({pic cs:m2}) to ({pic cs:em2});
    %\draw [->] %([yshift=.75pt]{pic cs:a}) -- ({pic cs:c});
\end{tikzpicture}
\end{figure}

Proof of Theorem~\ref{thm:lb1}: Consider the language $L_n = \{uv\: |\: u,v \in \{0,1\}^n \mbox{ and }
u = v\}$.  Observe that any (finite) automaton recognizing $L_n$ must
have $2^n$ states. This is because if there is an automaton $M$ with
$< 2^n$ states tha{}t recognizes $L_n$, then there are two strings $u_1
\neq u_2 \in \{0,1\}^n$ such that $M$ is in the same state after
reading $u_1$ and $u_2$. This means either $M$ accepts both $u_1u_1$
and $u_1u_2$ or rejects both $u_1u_1$ and $u_1u_2$, which contradicts
the fact that $M$ recognizes $L_n$. Thus, any one pass TM for $L_n$
must use space $n$.

We will essentially show how to ``reduce'' checking membership in $L_n$ to
checking WCP. The reduction uses constantly many threads, locks and variables.
Let us consider the special case of $n = 3$. Suppose the input to $L_3$ is $w
= b_0b_1b_2c_0c_1c_2$. For this input $w$, we will construct the trace shown
in Figure~\ref{fig:lowerbound} (which is a parameterized and extended version
of the trace in Figure~\ref{fig:algoex}), where in the trace, the lock
$b_i,c_j \in \{\ell_0,\ell_1\}$, depending on what the corresponding bit in
$w$ is. Once again the edges shown in the graph correspond to edges produced
by Rules (a) or (b), but in this case the edges are contingent on the values
of $b_i$s and $c_i$s. For example the edge from $\rel{b_0}$ to $\rel{c_0}$ and
the edge between $\rel{m}$ (call this edge $m_0$) on Lines 10 and 29 are
dependent on $b_0 = c_0$. Going forward the the edge from $\rel{b_1}$ to
$\rel{c_1}$ depends on the edge $m_0$ ($b_0 = c_0$) and $b_1 = c_1$. The
argument continues till you reach the edge between $\rel{m}$  between Lines 24
and 37 which requires $b_0b_1b_2 = c_0c_1c_2$. This implies the two \W{z}
events are WCP ordered iff $b_i = c_i$ for all $i$. Therefore checking whether
the word $w = b_0b_1b_2c_0c_1c_2$ boils down to checking whether there is a
WCP-race. 

Proof of Theorem~\ref{thm:lb2}: Note that the lower bound we initially proved in Section~\ref{sec:algo}  applies to algorithms that do a
single pass over the trace. But the above argument can in fact be generalized
to prove a combined time and space trade-off on any algorithm as follows. Note
that the communication complexity of checking if two $n$-bit strings are equal
is $\Omega(n)$. Consider the language $L_n = \{u\#^nv\: |\: u,v \in
\{0,1\}^n \mbox{ and } u = v\}$ and its membership problem. Now consider any
Turing Machine $M$ that is allowed to solve $L_n$ by going back and forth. If
$M$ takes $T(n)$ time then we know it has to make at most $\frac{T(n)}{n}$
``rounds'' of $\#^n$. If the space requirement of $M$ is $S(n)$ then we know
it can carry at most $S(n)$ bits in each rounds, implying it communicated
$\frac{T(n)S(n)}{n}$ bits across the channel of $\#^n$. The total number of
bits it needs to communicate in the end is $n$, which means $T(n)S(n) \in
\Omega(n^2)$. The words $u\#^nv$ can once again be modeled in a trace as
before with $\#^n$ being junk events in the trace and correspondence between
membership in $L_n$ and WCP race detection can be shown exactly in the same
way as before.

\end{document}